%% file: paper.tex
\documentclass[onefignum,onetabnum]{siamart190516}

\headers{Local Flow Partitioning for Faster Edge Connectivity}{M. Henzinger, S. Rao, and D. Wang}

\title{Local Flow Partitioning for Faster Edge Connectivity\thanks{The research leading to these results has received funding from the European Research Council 
under the European Union's Seventh Framework Programme (FP/2007-2013) / ERC Grant
Agreement no. 340506 and was done in part while Monika Henzinger was visiting the Simons Institute of the Theory of Computing. Satish Rao and Di Wang were supported by NSF Grant NSF Grant CCF-1528174 and CCF-1535989 during this work.}}

\author{Monika Henzinger\thanks{University of Vienna, 
  (\email{monika.henzinger@univie.ac.at}).}
\and Satish Rao\thanks{UC Berkeley 
  (\email{satishr@cs.berkeley.edu}).}
\and Di Wang\thanks{UC Berkeley 
  (\email{wangd@eecs.berkeley.edu}).}}

\usepackage{amsmath}
\usepackage{mathtools,amsfonts,enumerate,enumitem}
\usepackage{thmtools, thm-restate}

\theoremstyle{plain}
\theoremheaderfont{\normalfont\sc}
\theorembodyfont{\normalfont\itshape}
\theoremseparator{.}
\theoremsymbol{}
\newsiamthm{mylemma}{Lemma}
\newsiamthm{observation}{Observation}

\newcommand{\vol}{{\hbox{\bf vol}}}

\newcommand{\cc}{{0.6}}
\newcommand{\GG}{{\overline{G}}}
\newcommand{\core}{{4}}
\newcommand\myeq{\mathrel{\overset{\makebox[0pt]{\mbox{\normalfont\tiny\sffamily def}}}{=}}}

\newcommand{\tab}{.\hskip.1in}

\DeclarePairedDelimiter\ceil{\lceil}{\rceil}
\DeclarePairedDelimiter\floor{\lfloor}{\rfloor}
\newcommand{\ex}{\operatorname{ex}}
\newcommand{\ind}{\operatorname{in}}
\newcommand{\outd}{\operatorname{out}}

\begin{document}

\maketitle{}

\input{abstract}
\input{intro}

\input{prelim}

\input{section4}

\input{appendix-flow.tex}
\input{section2}

\input{appendix-sec2.tex}

\input{section3}

\input{appendix-inner.tex}

\input{runtime}

\bibliographystyle{siamplain}
\bibliography{min-cut}
%
%

\end{document}

%% file: abstract.tex
\begin{abstract} 
We study the problem of computing a minimum cut in a simple, undirected graph and give a deterministic  $O(m \log^2 n \log\log^2 n)$ time algorithm. This improves both on the best previously known deterministic running time of $O(m \log^{12} n)$ (Kawarabayashi and Thorup \cite{KT18}) and the best previously known randomized running time of $O(m \log^{3} n)$ (Karger \cite{KargerMinCut00}) for this problem, though Karger's algorithm can be further applied to weighted graphs. Moreover, our result extends to {\em balanced} directed graphs, where the {\em balance} of a directed graph captures how close the graph is to being Eulerian.

Our approach is using the Kawarabayashi and Thorup graph compression technique, which repeatedly finds low-conductance cuts. To find these cuts they use a diffusion-based local algorithm. We use instead a flow-based local algorithm and suitably adjust their framework to work with our flow-based subroutine. Both flow and diffusion based methods have a long history of being
applied to finding low conductance cuts. Diffusion algorithms have
several variants that are naturally local while it is more complicated to make flow methods local. Some prior work has proven
nice properties for local flow based algorithms with respect to
improving or cleaning up low conductance cuts. Our flow subroutine,
however, is  the first that is both local and produces low
conductance cuts. Thus, it may be of independent interest.
\end{abstract}

%% file: intro.tex
\section{Introduction}

Given an unweighted (or simple) graph $G= (V,E)$ with $n = |V|$ and $m = |E|$, the edge connectivity $\lambda$ of $G$ is the size of the smallest  edge set whose removal disconnects the graph.
Given an edge-weighted graph $G_w=(V,E,w)$ the minimum cut of $G_w$ is the weight of the minimum weight edge set whose removal disconnects the graph.
In a breakthrough paper in 1996, Karger~\cite{KargerMinCut00} gave the first randomized algorithm that computes the minimum cut in expected near-linear time and posed as an open question to find a deterministic  near-linear time minimum cut algorithm. Almost 20 years later, in a recent breakthrough,
Kawarabayashi and Thorup \cite{KT18} partially answered his open question, by presenting 
the first deterministic near-linear time algorithm for finding edge
connectivity in an unweighted simple graph.  They state their runtime as  $O(m \log^{12}n)$ .

Their contribution is on two levels. They improved the deterministic
runtime for edge connectivity to near linear time, and perhaps of more
or equal interest they developed a new deterministic algorithm that
computes from $G$ a sparser multi-graph $\overline{G}$ that preserves
all non-trivial minimum cuts in $G$, i.e., a deterministic sparsification of
$G$. We note that $\overline{G}$ is produced by a recursive procedure
and we refer to either $\overline{G}$ or the procedure as the {\em K-T
  decomposition.}
Applying Gabow's edge connectivity algorithm~\cite{Gabow91} (which
runs on multi-graphs) to $\overline{G}$ yields the claimed results for
computing edge connectivity.

We believe the K-T decomposition is of independent interest based on
the long line of research on sparsification and clustering and the
astounding impact in algorithms sparsification and clustering have
had. For example, the near-linear time solvers for linear systems
\cite{ST04,KMP,KOSZ} is one of the very important applications of
sparsification. This specific application as well as a large part of
the prior work on sparsification is based on randomization.  As the
K-T decomposition is deterministic and it introduces some quite
interesting ideas and structures, we believe it might well prove
useful in improving the state of the art with respect to the
co-evolution of graph decompositions and algorithms, and specifically
the use of deterministic sparsification in various algorithms.

A central tool used in \cite{KT18} to compute the K-T decomposition is
a local {\em probability mass diffusion method}, called a ``page
rank'' method.  We replace this diffusion method by a flow-based
method and modify the K-T algorithm to accomodate the differences
between these methods. As a result we derive an algorithm that has a
deterministic runtime of $O(m \log^2 n \log\log^2 n)$ for computing a
K-T decomposition and the edge connectivity in $G$. 
Note that our deterministic algorithm is faster than 
 the best known randomized edge connectivity algorithm, whose
 running time is
$O(m \log^3n)$~\cite{KargerMinCut00}.

Furthermore, the same approach also gives an efficient deterministic algorithm for simple directed graphs when a directed graph is {\em balanced}. The notion of {\em balance} is introduced in~\cite{EGJS16}, and it quantifies how close a directed graph is to being Eulerian. A directed graph is {\em $\alpha$-balanced} if for any cut $S$ in the graph, both the in-degree and out-degree of $S$ are at least $1/\alpha$ fraction of the sum of the in-degree and out-degree, i.e. the total number of edges across the cut $S$ (See Section~\ref{sxn:prelim} for formal definitions). Balanced directed graphs give natural transition from undirected graphs to directed graphs, as we can turn an undirected graph into a $2$-balanced directed graph by replacing each undirected edge by two directed edges in opposite directions, and the size of any cut remains the same. The efficiency of our algorithm is parameterized by the balance, and the running time is $O(\alpha^4 m \ln^{2}m\ln\ln^2 m)$ for an $\alpha$-balanced directed graph. 

In this paper, we present the more general algorithm for $\alpha$-balanced directed graphs rather than the algorithm for undirected graphs. A previous version of this paper appeared in SODA $2017$, and a slightly simpler algorithm that only works for undirected graphs was presented.  

{\bf Flows versus diffusion methods.}  From a technical point of view
our result contributes to the line of work on finding low conductance
cuts using local methods;
a {\em local method} being one whose
runtime depends only on the volume of the (smaller side of the) cut that it outputs. 
Flow and probability mass diffusions (or
more generally, spectral methods) have a long history of competing to
provide good graph decompositions. But diffusions have the upper hand
in terms of local methods, as the
fact that the diffusion process is a linear operator allows for the
detailed knowledge of its evolution, which then can be used to reason powerfully
about its behavior.  For example, Spielman and Teng \cite{ST04},
inspired by classical analysis of random walks by Lov{\'{a}}sz and
Simonovitz \cite{LS93}, showed that most vertices in a low conductance
cut are good starting points for a diffusion that finds a good cut.  
Past flow based methods, however, were subject
to the black-box of a flow algorithm that could adversarly send flow
in an inopportune direction. We overcome this difficulty by suitably 
modifying the flow computation and present
the first primarily flow-based local method for locating low
conductance cuts. We expect that our approach can used to 
speed up further conductance-based graph algorithms.

Previously, the methods of \cite{OSVV,KRV} combined the properties of
diffusions and flow algorithms to produce low conductance cuts. The
methods alternate between diffusions which find barely non-trivial
cuts in an embedded graph, and flow embedding edges to cross these
cuts. The flow computation interacts with the original graph without
the quadratic (in conductance) loss that is inherent in diffusion
methods\footnote{Indeed, one could also see the best known
  approximation algorithm for conductance in \cite{ARV} as such a
  combination.}.  Those methods, however, fail terribly to give local
methods for finding low conductance cuts, and explictly treat the flow
algorithm more as an constrained adversary rather a useful tool.

Our flow algorithm attempts to combine (some of) the power of
diffusion with the speed and efficiency of flow methods more tightly,
without actually using a diffusion method.  Basically it consists of
an excess scaling method repeatedly calling a modified push-relabel
algorithm.  The excess scaling portion enforces locality on our
algorithm and the details of our push-relabel method allow us to get
a handle on conductance. 

In recent work, some local flow based methods have been studied in a
similar vein. However, diffusions are still used when producing low
conductance cuts. For example, Orecchia and Zhu~\cite{OZ14} use a
detailed view of a blocking flow based flow algorithm to obtain
improved results on finding low conductance cuts; in particular, they
show how to locally find a $\tilde{O}(\frac{1}{\gamma})$ approximation
to conductance given a seed set overlapping the cut by a $\gamma$
fraction. They apply their method to local partitioning but use a
result of Allen, Lattanzi and Mirrokni \cite{ALM13} which in turn uses
diffusion or page-rank from~\cite{ACL06} (and as
does~\cite{KT18}).  

In the context of \cite{KT18}, our flow method roughly matches  diffusion
where it does well, and  dominates
diffusion with respect to its quadratic loss.  That is, the
decomposition developed in \cite{KT18}, repeatedly finds cuts of low
conductance $O(\frac{1}{\log^c n})$, or certifies a certain property
related to connectivity.  In \cite{KT18}, the local diffusions suffer
both a quadratic gap (as well as a logarithmic factor) between what
can be certified and the conductance of a cut as well as quadratic (in the conductance)
overhead in runtime. Our modified push-relabel algorithm
either certifies the property or finds a low conductance cut with only a logarithmic gap.
However, its runtime depends on the amount of ``source supply'' provided to it.
To make sure that this ``source supply'' is only $O(m \log n)$ we use the 
excess scaling procedure, which repeatedly calls the
push-relabel algorithm with suitably rescaled  source supply.
  This leads to the improvement in runtime for our
method. 

{\bf Other Previous Work.}
Work on edge connectivity and its generalization, the minimum cut
problem, has a long history perhaps beginning with Gomory and
Hu's \cite{GH61} use of the maximum flow problem to solve this problem.
Some relatively recent highlights include the work of Nagamochi
and Ibaraki \cite{NI90} which bypasses the use of the maximum
flow problem, and simple beautiful versions of these
by Frank \cite{Frank} and Stoer and Wagner \cite{SW}
which give $O(nm + n^2\log n)$ deterministic algorithms for minimum cut.

For edge connectivity of simple graphs, Gabow had the best previous
deterministic algorithm which was $O(m + \lambda^2 n)$ time
where $\lambda$ was the connectivity.  His methods could handle
parallel edges in $O(m + \lambda^2 n \log n)$ time.  Matula~\cite{Matula93}
has a linear time $(2+ \epsilon)$ approximation algorithm for this problem
as well. 

There is also substantial work in local graph partitioning including
the aforementioned work of Anderson, Chung, and Lang \cite{ACL06}
which gives a local diffusion process that outputs a set of conductance
$(\phi \log n)^{-1/2}$ in time $O(\phi^{-1} \log^c n)$ times
the size of the output for a good fraction of the starting
vertices in a cut of conductance $\phi$.  The runtime overhead was
improved to $\phi^{-1/2}$ using an evolving set diffusion
by Anderson and Peres \cite{AP}.  The heat kernel diffusion was 
used to improve the quality of the cut to $\phi^{-1/2}$ in \cite{C09},
though the impact on runtime overhead is not clear in that work. 
We note that the result of local diffusions have also had impact
empirically in, for example, the use of Personalized Page Pank \cite{PP15}.

{\bf Applications of our flow method beyond edge connectivity.} After this work was announced at SODA'17~\cite{HRW17}, the flow-based subroutine we developed has found further applications in other graph decomposition problems. In particular,~\cite{WangFHMR17} showed that an adaptation of our flow method has better theoretical and empirical behavior comparing to random walk based diffusion~\cite{ST08,ACL06,ZhuLM13} in the context of local graph clustering. Moreover, in the context of expander decomposition~\cite{KVV04}, an adaptation of our flow method was the key ingredient in the first nearly linear time algorithm to partition a graph into components with large internal conductance~\cite{SW19}, which improves the seminal result of Spielman and Teng~\cite{ST08}, where in their work the components as induced subgraph are not guaranteed to have large internal conductance. The stronger guarantee of~\cite{SW19} has led to further applications in the design of fast algorithms on graphs. 

{\bf Organization of Paper.}
In Section~\ref{sxn:prelim} we introduce the relevant definitions and notations. We present our flow procedure along
with its analysis in Section~\ref{sxn:flow}. We describe the
overall structure of the K-T decomposition in
Section~\ref{sxn:outerloop}, with some details deferred to the end of the section (Section~\ref{app:section2}).  We then present our version of the K-T
inner procedure in Section~\ref{sxn:inner-procedure}, and a detailed analysis in Section~\ref{app:inner-procedure}.  Finally,
Section~\ref{sxn:runtime} contains the running time analysis.  


%% file: prelim.tex
\section{Preliminaries and notations}
\label{sxn:prelim}
For an undirected graph $G=(V,E)$, we denote $d(v)$ as the degree $v$, and $\vol(C)$ as the {\em volume} of $C\subseteq V$, which is $\sum_{v\in C}d(v)$. A {\em cut} is a subset $S\subset V$, or $(S, \bar{S})$ where $\bar{S} = V\setminus S$, and $\bar{S},S \neq \emptyset$ .  The {\em cut-size} $\partial(S)$ of a cut $S$ is the number of edges between $S$ and $\bar{S}$. The {\em conductance} of a cut $S$ is
$\Phi(S)\myeq \frac{\partial(S)}{\min(\vol(S),\vol(V\setminus S))}$. Unless otherwise noted, when speaking of the
conductance of a cut $S$, we assume $S$ to be the side of minimum
volume.

For a directed graph $G=(V,E)$, We denote the in-degree of vertex $v$ by $\ind(v)$, and the out-degree of $v$ by $\outd(v)$. We extend the notations to cuts, where $\ind(S)$ denotes the number of edges from $\bar{S}$ to $S$, and $\outd(S)$ is the number of edges from $S$ to $\bar{S}$. Moreover, for disjoint subsets of nodes $S_1,S_2$, we denote $E(S_1,S_2)$ as the set of edges directed from a node in $S_1$ to a node in $S_2$, and when $S_1$ (or $S_2$) has only one vertex, we simply write the vertex instead of the singleton set. When we use the undirected notations on a directed graph, we simply ignore the directions of the edges, and look at the undirected version. For example, $d(v)$ for vertex $v$ in a directed graph is simply $\ind(v)+\outd(v)$, and $\partial(S)=\ind(S)+\outd(S)$ for cut $S$. A directed graph $G$ is {\em $\alpha$-balanced} if $\ind(S)$ and $\outd(S)$ are both at least $\partial(S)/\alpha$ for all cut $S$ of $G$.

An undirected graph is {\em simple} if there is no self-loop and parallel edge. A directed graph is {\em simple} if there is no self-loop, and for any pair of nodes $(u,v)$, there is at most one edge directed from $u$ to $v$. When we treat a simple directed graph as an undirected graph, there will be at most two parallel edges between any pair of nodes.

The {\em internal edges} of a set $C\subseteq V$ are the edges with both endpoints in $C$. We add $H$ or $C$ as subscripts, i.e. $d_H(v),\vol_C(A),\phi_C(A)$ etc., if we consider only the internal edges of a subgraph $H$ or a subset $C\subseteq V$, while we omit the subscripts when the graph is clear from context. We use $m$ to denote the number of (internal) edges of a graph, and again add a subscript to $m$ when there are multiple graphs in the context.

A cut $S$ is {\em non-trivial} if $|S|,|\bar{S}|>1$. Given $A\subset V$ and a cut $S$ we say that $A$ {\em contains} the cut $S$
if there exist nodes $u$ and $v$ in $A$ such that $u \in S$ and $v \in \bar{S}$. Otherwise, we say that $A$ {\em does not contain} the cut.

For an undirected graph $G$, a cut is a {\em non-trivial} minimum cut, if it is a cut of minimum cut-size in $G$ {\em and} it is non-trivial.\footnote{Thus, it is {\em not necessarily} the cut of minimum cut-size out of all non-trivial cuts.} 

For a directed graph $G$, a cut $S$ is a {\em non-trivial} minimum cut, if it is a cut of minimum number of outgoing edges (i.e., $\outd(S)$) among all cuts in $G$ {\em and} it is non-trivial. Minimum cut can be equivalently defined with respect to the number of incoming edges, and if $S$ is a non-trivial minimum cut with respect to the number of outgoing edges, then $\bar{S}$ is a non-trivial minimum cut with respect to the number of incoming edges. Thus, we refer to $(S,\bar{S})$ as a non-trivial minimum cut, if either $S$ or $\bar{S}$ is a non-trivial (outgoing) minimum cut. 

Since we can check the size of all trivial cuts by looking at the degree of nodes, our algorithm pretends the minimum cut is non-trivial. To get the minimum cut, we compare the size of the cut returned by the algorithm against the minimum trivial cut.

We will consider flow problems (in undirected graphs) extensively. Formally, a {\em flow problem} $\Pi$ is defined with a {\em source function}, $\Delta:V\rightarrow \mathbb{Z}_{\geq 0}$, a {\em sink function}, $T:V\rightarrow \mathbb{Z}_{\geq 0}$, and edge capacities $c(\cdot)$. We say that $v$ is {\em a sink of capacity} $x$ if $T(v)=x$. All flow problems we consider in this work use the same sink function, $\forall v: T(v)=d(v)$, so we won't explicitly write down $T(\cdot)$. To avoid confusion with the way flow is used, we use {\em supply} to refer to the substance being routed in flow problems.

For the sake of efficiency, we will not typically obtain a full solution to a flow problem. We will compute a pre-flow, which is a function $f:V\times V \rightarrow R$, where $f(u,v) = -f(v,u)$. A pre-flow $f$ is {\em source-feasible} with respect to source function $\Delta$ if $\forall v:\sum_u f(v,u) \leq \Delta(v)$. A pre-flow $f$ is {\em capacity-feasible} with respect to $c(\cdot)$ if $|f(u,v)| \leq c(e)$ for $e=\{u,v\} \in E$ and $f(u,v) = 0$ otherwise. We say that $f$ is a {\em  feasible pre-flow} for flow problem $\Pi$, or simply a {\em pre-flow} for $\Pi$, if $f$ is both source-feasible and capacity-feasible with respect to $\Pi$. 

For a pre-flow $f$ and a source function $\Delta(\cdot)$, we extend the notation to denote $f(v)\myeq\Delta(v)+\sum_u f(u,v)$ as {\em the amount of supply ending at $v$ after $f$}. Note that $f(v)$ is non-negative for all $v$ if $f$ is source-feasible. When we use a pre-flow as a function on vertices, we refer to the function $f(\cdot)$, and it will be clear from the context what $\Delta(\cdot)$ we are using. If in addition, $\forall v:f(v)\leq T(v)$, the pre-flow $f$ will be a {\em feasible flow (solution)} to the flow problem $\Pi$. 

We denote $\ex(v)\myeq\max(f(v)-T(v),0)$ as the excess supply at $v$, and we call the part of the supply below sink capacity, i.e., $\min(f(v),T(v))$, as the supply routed to the sink at $v$, or {\em absorbed} by $v$. We call the sum of all the supply absorbed by vertices, $\sum_v \min(f(v),T(v))$, the total supply routed to sinks. Finally, given a source function $\Delta(\cdot)$, we define $|\Delta(\cdot)|\myeq\sum_v\Delta(v)$ as the total amount of supply in the flow problem. Note the total amount of supply is preserved by any pre-flow routing, so $\sum_v f(v) =|\Delta(\cdot)|$ for any source-feasible pre-flow $f$.

%% file: section4.tex
\newcommand{\first}{\operatorname{first}}
\newcommand{\last}{\operatorname{last}}
\newcommand{\current}{\operatorname{current}}

\section{Flow Algorithm}
\label{sxn:flow}
The main tool used in~\cite{KT18} is a local diffusion method that finds low conductance cuts, we use a flow based local method instead, which we describe in this section. Its basic building block is a unit flow method, which is used as a subroutine by an excess scaling flow algorithm.  It produces either a pre-flow
routing most of the source supply to sinks or  a small conductance cut.

The unit flow method works on flow problems where $\forall v:\Delta(v)\leq wd(v)$ for constant $w\geq 2$. These flow problems are incremental in the sense that the initial excess supply on any $v$ is not too large compared to its sink capacity $d(v)$, so intuitively it requires limited work to spread the excess supply to sinks. Additionally, since the primary concern is to find low conductance cuts, instead of routing as much supply to sinks as possible, we use a 
Push-Relabel algorithm~\cite{GT14}, where we limit each label of a node to be at most a given parameter $h$, and we show
that at termination either ``enough'' flow was routed or  a low conductance cut with ``large enough'' volume can be found using a sweep cut method. These two aspects make the unit flow method very efficient.

We use excess scaling to divide a flow problem with a more general source supply function into multiple  incremental phases that it solves 
using the unit flow method. The basic idea is as follows: We use a parameter $\mu$, called {\em unit}, to scale down the source supply function such that a supply of $x$ turns into $x/\mu$ units, each  unit corresponding to a supply of
$\mu$. We choose the initial value $\mu$  large enough, so that after scaling down every $\Delta(v)$ by $\mu$ the source supply in unit $\mu$ satisfies $\forall v:\Delta(v)\leq 2d(v)$. Given the source supply function in unit $\mu$, the unit flow method either returns a low conductance cut $(A,\bar A)$, where $\min(\vol(A), \vol(\bar A))$ is ``large'', or it returns a flow that spreads out the supply so that a constant fraction of the total source supply is routed to vertices and each vertex $v$ receives  at most $d(v)$ units of supply. In the earlier case we terminate, in the later case we discard all source supply that we did not succeed in routing (and show that this only discards a constant fraction of the initial source supply in total) and then  scale down $\mu$ by $2$.
Thus, in the new unit value, each vertex $v$ has  at most $2d(v)$ units of supply, which we use as source supply for the next unit flow invocation. Note that when we work in unit $\mu$, the sink capacity of $v$ is $d(v)$ units, i.e. $d(v)\mu$  supply in unit 1. Thus when $\mu$ is large, vertices have and transfer large amount of supply,
which  limits the volume of the subgraph that the unit flow procedure needs to explore to either send flow to or to find a low conductance cut in. As we decrease the value of $\mu$ geometrically, successive invocations of
the  unit flow method 
explore larger and larger subgraphs. This allows us to terminate early when there is a low conductance cut of small volume, and is the key to achieve local runtime.

\input{unitFlow.tex}

\subsection{Excess Scaling Flow Algorithm}
\label{sxn:flow-algorithm}
The excess scaling procedure (Algorithm~\ref{alg:scaling-flow}) takes as input an undirected graph $G$ (with parallel edges) of volume $2m$, source function $\Delta$ such that $|\Delta(\cdot)|=2m$, constant $\tau\in (0,1)$, capacity parameter $U$,  and an integer $h\ge \ln m$. Recall that each vertex $v$ is a sink of capacity $d(v)$. The algorithm will either in time $O(mh)$  route at least $(1-\tau)2m$ supply to sinks, or find a low conductance cut $(K,\bar{K})$ in time proportional to $\min(\vol(K),\vol(\bar K))$. Formally, we have the following result, whose proof is in Section~\ref{app:excessscaling}.
\begin{restatable}{mylemma}{excessScaling}
\label{lemma:flow-alg}
Given a graph $G$ of volume $2m$, a source function $\Delta$ such that $|\Delta(\cdot)|=2m$, a  constant $0<\tau<1$, and positive parameters  $U$ and $h$, the flow procedure will return a preflow $f$, subject to edge capacity of $2UF$ on every edge, where $F=\max_v \frac{\Delta(v)}{2d(v)}$. It will also return $\Delta'(\cdot)$, the amount of source supply from each vertex that is routed to sinks, where each $v$ is a sink of capacity $d(v)$. In addition, we have either of the two cases below:
\begin{enumerate}[label={(\arabic*)}]
\item At least a $(1-\tau)$ fraction of the total source supply is routed to sinks
\[
|\Delta'(\cdot)|\geq (1-\tau) 2m
\]
The running time is $O(mh)$ in this case.
\item It returns a cut $(K,\bar{K})$, $\vol(K)\leq \vol(\bar{K})$, and $\vol(K)$ is $\Omega(\frac{m}{F\ln m\ln\ln m})$. The running time is\\ $O(\vol(K)h\ln\frac{m}{\vol(K)}\ln\ln m)$. Furthermore
\begin{enumerate}
\item If $h\ge \ln m$, $\Phi(K)\leq \frac{20\ln 2m}{h}+\frac{2}{U}$.
\item If $h=\Omega(\ln m'\ln\ln m')$ with $m'\geq m$, $\Phi(K)\leq \frac{(\log m+1-\ceil*{\log \vol(K)})}{20\log m'}+\frac{2}{U}$  
\end{enumerate}
\end{enumerate}
\end{restatable}
The procedure divides the flow problem into incremental phases, and uses successive {\em Unit-Flow} invocations on them. This is done via a parameter $\mu$, which is the value of one unit in {\em Unit-Flow}.  Initially,  $\mu=\max_v \frac{\Delta(v)}{2d(v)}$ such that each $v$ has initial source supply at most $2d(v)$ units.
It then calls {\em Unit-Flow} with scaled source function $\Delta/\mu$ and $w=2$. Every unit of supply in {\em Unit-Flow} is supply of value $\mu$ in the original problem. To avoid confusion, when we say {\em $x$ supply}, we mean a supply of value $x$, and when we say {\em $x$ units of supply}, we mean a supply of value $x\mu$.
Algorithm~\ref{alg:scaling-flow} calls {\em Unit-Flow} repeatedly with a geometrically decreasing value of $\mu$. The sink capacity of $v$ is $d(v)$ units in {\em Unit-Flow}, but the pre-flow returned by {\em Unit-Flow} may have excess supply on vertices. To use the supply on vertices at the end of a {\em Unit-Flow} invocation as the source supply of the next {\em Unit-Flow} call, we simply discard all excess supply (as we show this will only discard a small fraction of the total supply). Then 
there is at most $d(v)$ supply in unit $\mu$ at each vertex $v$. 
Thus we can halve the value of $\mu$ so that each $v$ has at most $2d(v)$ supply in the new unit. If, however, every node $v$ has
at most $d(v)$ supply in unit 1, we terminate as each vertex can absorb its supply.

From a flow point of view in the $j$-th call to {\em Unit-Flow} for $j = 0, 1, ...$ each node $v$ has a source supply $\Delta_j(v)$, 
where $\Delta_0(v) = \Delta(v)$ and for $j > 0$, $\Delta_j(v) = \mu \cdot \min(d(v),f_{j-1}(v))$ (the min captures the removal of excess supply), where $f_{j-1}(v)$ is the amount of supply ending at $v$ after the $j-1$-st call to {\em Unit-Flow}. Assume for the
moment that $f_{j-1}(v) \le d(v)$. Then for $j > 0$, $\Delta_j(v) = \mu \cdot f_{j-1}(v)$, 
i.e., each node $v$ has as source supply in the $j$-th call to {\em Unit-Flow} exactly the supply values that it received in the previous call. Thus, no
supply is absorbed at nodes between consecutive calls of {\em Unit-Flow}, the supply is just ``spread out'' more and more. Once the supply
ending at each node is at most its degree, the procedure terminates. Due to the removal of excess supply this happens for sure when
$\mu = 1$, but it might already happen for a larger value of $\mu$. As the 
 final flow $f$ is the sum of all flows $f_j$ and each call to {\em Unit-Flow} uses at most $U \mu$ edge capacity with
$\mu$ geometrically decreasing, each edge carries at most 
$2 U \max_v \frac{\Delta(v)}{2d(v)}$ flow. As the total source supply given to the $j$-th call is $|\Delta_j(\cdot)/\mu| \le 2m/\mu$, its runtime is
$O(mh/\mu)$ and as $\mu$ decreases geometrically the total time for all calls to {\em Unit-Flow} is  $O(mh/\mu_f)$, where $\mu_f$ is the value of $\mu$ at termination.

Algorithm~\ref{alg:scaling-flow} returns a pre-flow $f$, a possibly empty cut $A$, and a function $\Delta'(\cdot)$  on vertices 
 such that $\Delta'(v)$
is the amount of the $\Delta(v)$ source supply starting at $v$ that is routed to sinks at the end, i.e. never removed as excess supply. Since we can mark the supply with the original source vertex, and the invocations of {\em Unit-Flow} maintain the marks, $\Delta'(\cdot)$ will be easy to compute.
\noindent
\begin{algorithm}
\caption{Excess scaling flow procedure}
\label{alg:scaling-flow}
\fbox{
\parbox{0.97\textwidth}{
\tab {\bf Input}: $G=(V,E)$, $\Delta(\cdot)$, $\tau$, $U$, $h$.\\
\tab {\bf Initialization}: Let $F=\max_v \frac{\Delta(v)}{2d(v)}$, $\mu=F$, $j=0$, $\Delta_0 = \Delta'=\Delta$, 
\phantom{\tab  {\bf Initialization}: } $f$ be zero pre-flow\\
\tab {\bf Repeat}\\
\tab \tab {\em Note:} $\Delta_j(v)\leq 2d(v)\mu \quad \forall v$\\
\tab \tab Run {\em Unit-Flow}$(G,\frac{\Delta_j(v)}{\mu},U, h, w=2)$, and get back $f_j$ in unit $\mu$, and $A_j$.\\
\tab \tab {\em Add $f_j$ to our current preflow:}\\
\phantom{\tab\tab}  $ f(v,u)\leftarrow f(v,u)+f_j(v,u)\mu, \forall(v,u)$.\\
\tab \tab {\em Remove excess supply on vertices:} \\
\phantom{\tab\tab} $\Delta_{j+1}(v)=(f_j(v)-\ex_j(v))\mu, \forall v$. Update $\Delta'$.\\
\phantom{\tab\tab} (Recall $\Delta'(v)$ is the amount of source supply starting at $v$ that hasn't been removed as excess.)
\tab \tab {\bf If} $\vol(A_j)\geq \frac{\tau 2m}{10\mu\ln 2\mu \ln \ln m}$:\\ 
\phantom{\tab \tab}{\bf Return} $f$, $\Delta'$, and $K\myeq$ smaller side of $A_j,\bar{A_j}$. {\bf Terminate.}\\
\tab \tab {\bf If} $\forall v: \Delta_{j+1}(v) \leq d(v)$:\\
\phantom{\tab \tab} {\bf Return} $f$, $\Delta'$, and $K\myeq\emptyset$. {\bf Terminate.} // Supply at most $d(v) \forall v$\\
\tab \tab $\mu\leftarrow \mu/2$, $j\leftarrow j+1$, proceed to next iteration.\\
\tab {\bf End Repeat} 
}}
\end{algorithm}
Each {\em Unit-Flow} invocation returns a
possibly empty  low conductance cut. If at any point the volume of the returned cut is large compared to the total work done so far, the algorithm can terminate with a low conductance cut $(K, \bar K)$  in time $\tilde O(h \vol(K))$, i.e., in `` local'' time. If this never happens, since the volume of the cut returned after each {\em Unit-Flow} upper-bounds the amount of removed excess supply (Observation~\ref{obs:excess}),  the algorithm must route at least $(1-\tau)2m$ supply to sinks at the end. 

%% file: unitFlow.tex
\subsection{Unit Flow}
\label{sxn:unit-flow}
The {\em Unit-Flow} subroutine (Algorithm~\ref{alg:unit-flow})
 takes as input an undirected graph $G=(V,E)$ (with parallel edges but no self-loops), source function
$\Delta$ and integer $w\geq 2$ such that $\forall v:0\leq \Delta(v)\leq wd(v)$, as well as an integer capacity $U>0$ on all edges. Each vertex $v$ is a sink of capacity $d(v)$. Furthermore, the procedure takes as input an integer $h\ge\ln(|E|)$ to customize the push-relabel algorithm, which we describe next.

In our push-relabel algorithm, each vertex $v$ has a non-negative integer label $l(v)$ which is initially zero. The label of a vertex only increases during the execution of the algorithm and (in a modification of the standard push-relabel technique) 
cannot become larger than $h$. The bound of $h$ on the labels makes the runtime of {\em Unit-Flow} linear in $h$, but it may prevent our algorithm from routing all units of supply to sinks even when there exists a feasible routing for the flow problem. However, when our algorithm cannot route a feasible flow, allowing labels of value up to $h$ is sufficient to find a cut with low conductance (i.e.,~of value inversely proportional to $h$), which is our primary concern. 

The algorithm maintains a pre-flow and the standard residual network, where each undirected edge $\{v,u\}$ in $G$ corresponds to two directed arcs $(v,u)$ and $(u,v)$, with flow values such that $f(v,u)=-f(u,v)$, and $|f(v,u)|,|f(u,v)|\leq U$. The residual capacity of an arc $(v,u)$ is $r_f(v,u)=U-f(v,u)$. We also maintain $f(v)=\Delta(v)+\sum_u f(u,v)$, which will be non-negative for all nodes $v$ during the execution. The algorithm will explicitly enforce $f(v)\leq wd(v)$ for all $v$ through the execution (i.e., it does not push flow to a vertex $v$ if that
would result in $f(v) > w d(v)$).
\begin{algorithm}
\caption{Unit Flow}
\label{alg:unit-flow}
\fbox{
\parbox{0.97\textwidth}{
{\em Unit-Flow}($G$,$\Delta$,$U$,$h$,$w$) \\ 
\tab {\bf Initialization:}\\
\tab \tab $\forall \{v,u\}\in E$, $f(u,v)=f(v,u)=0$, $Q=\{v|\Delta(v)>d(v)\}$.\\
\tab \tab $\forall v$, $l(v)=0$, and $\current(v)$ is the first edge in its list of incident edges.\\
\tab {\bf While} $Q$ is not empty\\
\tab \tab Let $v$ be the first vertex in $Q$, i.e. the lowest labelled active vertex.\\
\tab \tab {\em Push/Relabel}$(v)$. \\
\tab \tab {\bf If} {\em Push/Relabel}$(v)$ pushes supply along $(v,u)$\\
\tab \tab \tab {\bf If} $u$ becomes active, {\em Add}$(u,Q)$\\
\tab \tab \tab {\bf If} $v$ becomes in-active, {\em Remove}$(v,Q)$\\
\tab \tab {\bf Else If} {\em Push/Relabel}$(v)$ increases $l(v)$ by $1$\\
\tab \tab \tab {\bf If} $l(v)<h$, {\em Shift}$(v,Q)$\footnote{Recall $Q$ maintains the set of active vertices by their labels in non-decreasing order. In this case, we raised the label of $v$ by $1$, and {\em Shift}$(v,Q)$ shifts the position of $v$ in $Q$. Note {\em Add, Remove} operation always take place at the head of $Q$, and {\em Shift} always deals with the case when the label is raised by $1$, so we can implement $Q$ using linked lists, such that {\em Shift, Add, Remove} all take $O(1)$ time per operation.}, {\bf Else} {\em Remove}$(v,Q)$\\
\tab \tab {\bf End If}\\
\tab {\bf End While}}}
\fbox{
\parbox{0.97\textwidth}{
{\em Push/Relabel}$(v)$ \\
\tab Let $\{v,u\}$ be $\current(v)$.\\
\tab {\bf If} {\em Push}$(v,u)$ is applicable, then {\em Push}$(v,u)$.\\
\tab {\bf Else}\\
\tab \tab {\bf If} $\{v,u\}$ is not the last edge in $v$'s list of edges.\\
\tab \tab \tab Set $\current(v)$ to be the next edge in $v$'s list of edges.\\
\tab \tab {\bf Else} (i.e. $(v,u)$ is the last edge of $v$)\\
\tab \tab \tab {\em Relabel}$(v)$, and set $\current(v)$ to be the first edge of $v$'s list of edges.\\
\tab \tab {\bf End If}\\
\tab {\bf End If}
}}
\fbox{
\parbox{0.97\textwidth}{
{\em Push}$(v,u)$\\
\tab {\bf Applicability}: $\ex(v)>0,r_f(v,u)>0$, $l(v)=l(u)+1$.\\
\tab {\bf Assertion}: $f(u)<wd(u)$.\\
\tab $\psi = \min\left(\ex(v),r_f(v,u),wd(u)-f(u)\right)$\\
\tab Send $\psi$ units of supply from $v$ to $u$:\\
\phantom{\tab} $f(v,u)\leftarrow f(v,u)+\psi, f(u,v)\leftarrow f(u,v)-\psi$.
}}
\fbox{
\parbox{0.97\textwidth}{
{\em Relabel}$(v)$\\
\tab {\bf Assertion}: $v$ is active, and $\forall u\in V$, $r_f(v,u)>0\implies l(v)\leq l(u)$.\\
\tab $l(v)\leftarrow l(v)+1$.
}}
\end{algorithm}
As in the generic push-relabel framework, an {\em eligible} arc $(v,u)$ is a
pair such that $r_f(v,u)>0$ and $l(v) = l(u)+1$. A vertex $v$ is {\em active} if $l(v)<h$ and $\ex(v)>0$. The
algorithm  maintains the property that for any arc $(v,u)$ with
positive residual capacity,  $l(v) \leq l(u) + 1$.
The algorithm repeatedly picks an active vertex $v$ with minimum label and either pushes along an eligible arc incident to
$v$ if there is one, or it raises the label of $v$ by $1$ if there is no eligible arc out of $v$.

Upon termination of the algorithm, we will have a pre-flow $f$, as well as labels $l(\cdot)$ on vertices. {\em Unit-Flow} will either successfully route a large amount of supply to sinks, or we can compute a low conductance cut using the labels. The proof is in Section~\ref{app:unitflow}. 
\begin{restatable}{theorem}{unitFlow}
\label{thm:unit-flow}
Given $G,\Delta,h,U,w\geq 2$ such that $\Delta(v)\leq wd(v)$ for all $v$, {\em Unit-Flow} terminates with a pre-flow $f$, where we must have one of the following three cases
\begin{enumerate}[label={(\arabic*)}]
\item $f$ is a feasible flow, i.e. $\forall v:\ex(v)=0$. All units of supply are absorbed by sinks.
\item $f$ is not a feasible flow, but $\forall v:f(v)\geq d(v)$, i.e., at least $2m$ units of supply are absorbed by sinks.
\item If $f$ satisfies neither of the two cases above, we can find a cut $(A,\bar{A})$ such that $wd(v)\geq f(v)\geq d(v)$ for all $v\in A$, and $f(v)\leq d(v)$ for all $v\in \bar{A}$. Furthermore
\begin{enumerate}
\item If $h \ge \ln m$, the conductance $\Phi(A)=\frac{|E(A,V\setminus A)|}{\min(\vol(A),2m-\vol(A))}\leq \frac{20\ln 2m}{h}+\frac{w}{U}$.
\item If $h=\Omega(\ln m' \ln\ln m')$ for $m'\geq m$, we have a more fine-grained conductance guarantee: let $K$ be the smaller side of $(A,\bar{A})$, then $\Phi(K)\leq \frac{\ln m + 1 - \ceil*{\ln \vol(K)}}{50\ln m'}+\frac{w}{U}$.
\end{enumerate}
\end{enumerate}
\end{restatable}
The motivation for maintaining $f(v)\leq wd(v)$ throughout the algorithm is to establish lower bounds on $\vol(A)$. Intuitively, if the total amount of excess supply is large at the end, $\vol(A)$ must be large as no single vertex can have too much excess. 
More specifically, we have the following observations.
\begin{observation}
\label{obs:excess}
If the output fulfills case $(3)$ of Theorem~\ref{thm:unit-flow}, we  have 
\[
\sum_{v\in V}\ex(v)\leq (w-1)\vol(A)
\]
\end{observation}
\begin{observation}
\label{obs:largevol}
When $|\Delta(\cdot)|\geq tm$ for constant $t>2$, for any pre-flow $f$ we must have 
\[
\sum_{v\in V}\ex(v)\geq (t-2)m.
\] 
If $w$ is a constant, and we get case $(3)$, then every node $v$ in $A$ absorbs
$d(v)$ units of supply, and $\vol(A)=\Theta(m)$.
\end{observation}

{\em Unit-Flow} returns a pre-flow $f$ and a possibly empty cut $A$.
Additionally, we treat units of supply as distinct tokens with marks bearing information, which must be preserved by the pre-flow, leading to an extra $O(1)$ work per push of a single
unit.
This leads to the following running time result, whose proof is in Section~\ref{app:unitflow}.
\begin{restatable}{mylemma}{unitFlowTime}
\label{lemma:unitflow-runtime}
The running time for {\em Unit-Flow} is $O(w|\Delta(\cdot)|h)$. The running time includes 
identifying which case of Theorem~\ref{thm:unit-flow} we end in, and in case $(3)$, identifying 
the set A.
\end{restatable}

%% file: appendix-flow.tex
\subsection{Analysis of {\em Unit-Flow}}
\label{app:unitflow}
Recall the {\em Unit-Flow} procudure (Algorithm~\ref{alg:unit-flow}) is a fairly straightforward implementation of the push-relabel framework, with some notable design decisions as follows:
\begin{itemize}
\item We explicitly maintain upperbounds on the supply remaining at a vertex, i.e. $f(v)\leq wd(v)$ when we push to $v$. We assume this holds at the start, i.e. the input $\Delta(v)\leq wd(v)$ for all $v$.
\item We cap the labels at $h$. If a vertex has label $h-1$, and we relabel it to $h$, the vertex never becomes active from then on.
\item The active vertices in $Q$ are in non-decreasing order with respect to their labels, and each time we need to get an active vertex from $Q$, we get the first vertex.
\end{itemize}
Note that the assertion in {\em Push}$(v,u)$ is the reason we always use the active vertex $v$ with the smallest label. If {\em Push}$(v,u)$ can be applied, we must have $f(u)\leq wd(u)$. This is because we know $l(v)=l(u)+1$, so $l(u)<h$, and if $f(u)\geq wd(u)$, then $u$ has positive excess. Thus, $u$ is active in this case, which contradicts $v$ being the active vertex with the smallest label. The applicability conditions and the assertion guarantee that we can push $\psi\geq 1$ unit of supply from $v$ to $u$. 


Upon termination, we have a pre-flow $f$, and labels $l$ on vertices. We make the following observations.
\begin{observation} 
\label{obs:unit-flow}
During the execution of {\em Unit-Flow}, we have
\begin{enumerate}[label={(\arabic*)}]
\item If $v$ is active at any point, the final label of $v$ cannot be $0$. The reason is that $v$ will remain active until either $l(v)$ is increased to $h$, or its excess is pushed out of $v$, which is applicable only when $l(v)$ is larger than $0$ at the time of the push.
\item Each vertex $v$ is a sink that can absorb up to $d(v)$ units of supply, so we call the $f(v)-\ex(v)=\min(f(v),d(v))$ units of supply remaining at $v$ the {\em absorbed} supply. The amount of absorbed supply at $v$ is in $[0,d(v)]$, and is non-decreasing. Thus any time after the point that $v$ first becomes active, the amount of absorbed supply is $d(v)$. In particular any time the algorithm relabels $v$, there have been $d(v)$ units of supply absorbed by $v$.
\end{enumerate}
 Upon termination of {\em Unit-Flow} procedure, we have
\begin{enumerate}[label={(\arabic*)}]
\setcounter{enumi}{2}
\item For any edge $\{v,u\}\in E$, if the labels of the two endpoints differ by more than $1$, say $l(v)-l(u)>1$, then arc $(v,u)$ is saturated. This follows directly from a standard property of the push-relabel framework, where $r_f(v,u)>0$ implies $l(v)\leq l(u)+1$. 
\end{enumerate}
\end{observation}
Although {\em Unit-Flow} may terminate with $\ex(v)>0$ for some $v$, we know all such vertices must have label $h$, as the algorithm only stops trying to route $v$'s excess supply to sinks when $v$ reaches level $h$. Thus we have the following lemma.
\begin{mylemma}
\label{lemma:excess}
Upon termination of {\em Unit-Flow} with input $(G,\Delta,U,h,w)$, assuming $\Delta(v)\leq wd(v)$ for all $v$, the pre-flow and labels satisfy
\begin{enumerate}[label={(\alph*)}]
\item If $l(v)=h$, $w d(v)\geq f(v)\geq d(v)$; 
\item If $h-1\geq l(v)\geq 1$, $w d(v) \geq f(v)=d(v)$; 
\item If $l(v)=0$, $f(v)\leq d(v)$.
\end{enumerate}
\end{mylemma}
\begin{proof}
By Observation~\ref{obs:unit-flow}.$(2)$, any vertex with label larger than $0$ must have $f(v)\geq d(v)$. The algorithm terminates when there is no active vertices, i.e. $\ex(v)>0\implies l(v)=h$, so all vertices with label below $h$ must have $f(v)\leq d(v)$. Moreover, $f(v)\leq wd(v)$ since at the beginning $f(v)=\Delta(v)\leq wd(v)$, and the push operations explicitly enforces $f(v)\leq wd(v)$ when pushing supply to $v$.
\end{proof}
Now we can prove the main result about {\em Unit-Flow}.\\
\unitFlow*
\begin{proof}
We use the labels at the end of {\em Unit-Flow} to divide the vertices into groups
\[
B_i=\{v|l(v)=i\}
\]
If $B_h=\emptyset$, no vertex has positive excess, so all $|\Delta(\cdot)|$ units of supply are absorbed by sinks, and we end up with case $(1)$.

If $B_h\neq \emptyset$, but $B_0=\emptyset$, by Lemma~\ref{lemma:excess} every vertex $v$ has $f(v)\geq d(v)$, so we have $\sum_vd(v)=2m$ units of supply absorbed by sinks, and we end up with case $(2)$.

{\bf Case (3a):} When both $B_h$ and $B_0$ are non-empty, we compute the cut $(A,V\setminus A)$ as follows: Let $S_i=\cup_{j=i}^h B_j$ be the set of vertices with labels at least $i$. We sweep from $h$ to $1$, and let $A$ be the first $i$ such that $\Phi(S_i)\leq 20(\frac{\ln 2m}{h}+\frac{w}{U})$. The properties $\forall v\in A: wd(v)\geq f(v)\geq d(v)$, and $\forall v\in V\setminus A: f(v)\leq d(v)$ follow directly from $S_h\subseteq A\subseteq S_1$. We will show that there must exists some $S_i$ satisfying the conductance bound. 

For any $i$, an edge $\{v,u\}$ across the cut $S_i$, with $v\in S_i,u\in V\setminus S_i$, must be one of the two types:
\begin{enumerate}
\item In the residual network, the arc $(v,u)$ has positive residual capacity $r_f(v,u)>0$, so $l(v)\leq l(u)+1$. But we also know $l(v)\geq i > l(u)$ as $v\in S_i,u\in V\setminus S_i$, so we must have $l(v)=i,l(u)=i-1$.
\item In the residual network, if $r_f(v,u)=0$, then $(v,u)$ is a saturated arc sending $U$ units of supply from $S_i$ to $V\setminus S_i$.
\end{enumerate}
Suppose there are $z_1(i)$ edges of the first type, and $z_2(i)$ edges of the second type. By the following region growing argument, we can show there exists some choice of $i=i^*$, such that 
\begin{equation}\label{eq:eligible}
z_1(i^*)\leq \frac{10\min(\vol(S_{i^*}),2m-\vol(S_{i^*}))\ln m}{h}
\end{equation}
If $\vol(S_{\floor*{h/2}})\leq m$, we start the region growing argument from $i=h$ down to $\floor*{h/2}$. By contradiction, suppose $z_1(i)\geq \frac{10\vol(S_i)\ln m}{h}$ for all $h\geq i\geq \floor*{h/2}$, which implies $\vol(S_i)\geq \vol(S_{i+1})(1+\frac{10\ln m}{h})$ for all $h>i\geq \floor*{h/2}$. Since $\vol(S_h)=\vol(B_h)\geq 1$ and $h\ge \ln m$, we will have $\vol(S_{\floor*{h/2}})\geq(1+\frac{10\ln m}{h})^{h/2}\gg 2m$, which gives contradiction. The case where $\vol(S_{\floor*{h/2}})> m$ is symmetric, and we run the region growing argument from $i=1$ up to $\floor*{h/2}$ instead. 

For any $i$, we can bound $z_2(i)$ as follows. Since the pre-flow pushes $z_2(i)U$ units of supply from $S_i$ to $V\setminus S_i$ along the $z_2(i)$ saturated arcs, $z_2(i)U$ is at most $\sum_{v\in S_i}\Delta(v)+z_1(i)U$, i.e. the sum of the source supply in $S_i$ and the supply pushed into $S_i$ along the $z_1(i)$ eligible arcs. As $\Delta(v)\leq wd(v)$ for all $v$, we know 
\[
z_2(i)\leq \frac{w\vol(S_i)}{U}+z_1(i)
\]
On the other hand, $z_2(i)U$ is at most $\sum_{v\in V\setminus S_i}f(v)+z_1(i)U$, as the $z_2(i)U$ units of supply pushed into $V\setminus S_i$ either remain at vertices in $V\setminus S_i$, or back to $S_i$ along the reverse arcs of the $z_1(i)$ eligible arcs. Since any $v\in V\setminus S_i$ is not with label $h$, thus $f(v)\leq d(v)$, then we get
\[
z_2(i)\leq \frac{\vol(V\setminus S_i)}{U}+z_1(i)
\]
The two upperbounds of $z_2(i)$ together give
\begin{equation}\label{eq:saturated}
z_2(i)\leq \frac{w\min(\vol(S_i),2m-\vol(S_i))}{U}+z_1(i)
\end{equation}

We know there exists some $i^*$ such that $z_1(i^*)$ is bounded by~\eqref{eq:eligible}, together with~\eqref{eq:saturated}, we have
\[
z_1(i^*)+z_2(i^*)\leq \min(\vol(S_{i^*}),2m-\vol(S_{i^*}))(\frac{20\ln m}{h} +\frac{w}{U})
\] 
thus $\Phi(S_{i^*})\leq \frac{20 \ln m}{h}+\frac{w}{U}$, which completes the proof.

{\bf Case (3b):} The proof is basically the same as the above case, but with a more careful region growing argument. In particular, we want to show there exists some $i^*,j^*$ such that $\min(\vol(S_{i^*}),2m-\vol(S_{i^*}))\leq 2^{j^*}$ and 
\begin{equation}
\label{eq:z1}
z_1(i^*)\leq \frac{\vol(S_{i^*})(\log m +1 -j^*)}{100\log m'}
\end{equation}
Assume $\vol(S_{\floor*{h/2}})\leq m$, and we run the region growing argument from $i=h$ down to $i=\floor*{h/2}$. In this case $\vol(S_i)\leq m$. The other case is similar, and we just do the region growing argument from the other side. 

Consider the groups $S_h,\ldots,S_{\floor*{h/2}}$, if we put the $\Omega(\ln m'\ln\ln m')$ groups into levels $j=1,\ldots,\log m$, such that a group $i$ is in level $j$ if $2^{j-1}\leq \vol(S_i)\leq 2^{j}$. There must be a level $j$ that gets more than $200\frac{\log m'}{\log m + 1 - j}$ groups, as long as $h>c_2\ln m'\ln\ln m'$ for some large constant $c_2$, since
\begin{align*}
\sum_{j=0}^{\log m} \frac{200\log m'}{\log m + 1 - j}&=200\log m'(1+\frac{1}{2}+\cdots+\frac{1}{\log m+1})\\
&\leq 500\ln m' \ln\ln m'
\end{align*}
Suppose this is level $j^*$, and let $i_a$ be the largest $i$ with $S_i$ in level $j^*$, and $i_b$ be the smallest $i$ with $S_i$ in level $j^*$. We know $i_a-i_b\geq 200\frac{\log m'}{\log m+1 -j^*}$, and $j^*-1\leq \log \vol(S_{i_a})\leq \log \vol(S_{i_b})\leq j^*$, thus there must be a $i^*\in [i_a,i_b]$ satisfying~\eqref{eq:z1}, since otherwise $\vol(S_{i_b})/\vol(S_{i_a})\geq (1+\frac{\log m +1 - j^*}{\log m'})^{200\frac{\log m'}{\log m+1 - j^*}}\gg 2$. Everything else follow the same arguments as in the proof of case $(3a)$ above.
\end{proof}
We proceed to prove the runtime of {\em Unit-Flow}. Recall we treat units of supply as distinct tokens, so a push operation of $\psi$ units takes $O(\psi)$ work to maintain the marks.\\
\unitFlowTime*
\begin{proof}
With a compact representation of $\Delta$, the initialization of $f(v)$'s and $Q$ takes time linear in $|\Delta(\cdot)|$. For the subsequent work, we will first charge the operations in each iteration of {\em Unit-Flow} to either a push or a relabel. Then we will in turn charge the work of pushes and relabels to the absorbed supply, so that each unit of absorbed supply gets charged $O(wh)$ work. This will prove the result, as there are at most $|\Delta(\cdot)|$ units of (absorbed) supply in total.

In each iteration of {\em Unit-Flow}, we look at the first element $v$ of $Q$, which is an active vertex with the smallest label. Suppose $l(v)=i$ at that point. If the call to {\em Push/Relabel}$(v)$ ends with a push of $\psi$ units of supply, the iteration takes $O(\psi)$ total work and we charge the work of the iteration to that push operation. If the call to {\em Push/Relabel}$(v)$ doesn't push, we charge the $O(1)$ work of the iteration to the relabel of $l(v)$ to $i+1$. If there is no such relabel, i.e. $i$ is the final value of $l(v)$, we know $i\neq 0$ by Observation~\ref{obs:unit-flow}$(1)$, then we charge the work to the final relabel of $v$. Since a relabel of $v$ must be incurred when $d(v)$ consecutive calls to {\em Push/Relabel}$(v)$ end with non-push, each relabel of $v$ takes $O(d(v))$ work by our charging scheme above.

So far we have charged all the work to pushes and relabels, such that pushing $\psi$ units of supply takes $O(\psi)$, and each relabel takes $O(d(v))$. We now charge the work of pushes and relabels to the absorbed supply. We consider the absorbed supply at $v$ as the first up to $d(v)$ units of supply starting at or pushed into $v$, and these units never leave $v$.

By Observation~\ref{obs:unit-flow}$(2)$, each time we relabel $v$, there are $d(v)$ units of absorbed supply at $v$, so we charge the $O(d(v))$ work of the relabel to the absorbed supply, and each unit gets charged $O(1)$. A vertex $v$ is relabeled at most $h$ times, so each unit of absorbed supply is charged with $O(h)$ in total by all the relabels.

For the pushes, we consider the potential function 
\[
\Lambda=\sum_v \ex(v)l(v)
\]
Each push operation of $\psi$ units of supply decrease $\Lambda$ by exactly $\psi$, since $\psi$ units of excess supply is pushed from a vertex with label $i$ to a vertex with label $i-1$. $\Lambda$ is always non-negative, and it only increases when we relabel some vertex $v$ with $\ex(v)>0$. When we relabel $v$, $\Lambda$ is increased by $\ex(v)$. Since $\ex(v)\leq f(v)\leq wd(v)$, we can charge the increase of $\Lambda$ to the absorbed supply at $v$, and each unit gets charged with $O(w)$. In total we can charge all pushes (via $\Lambda$) to absorbed supply, and each unit is charged with $O(wh)$. 

If we need to compute the cut $A$ as in case $(3)$ of Theorem~\ref{thm:unit-flow}, the runtime is $O(\vol(S_1))$. Recall $S_1$ is the set of vertices with label at least $1$, thus all with $d(v)$ units of absorbed supply, so $\vol(S_1)$ is at most $|\Delta(\cdot)|$.
\end{proof}

\subsection{Analysis of excess scaling procedure} 
Now we prove the main result of the excess scaling procedure, which is restated below.
\label{app:excessscaling}
\excessScaling*
\begin{proof}
Consider the call to {\em Unit-Flow} in one iteration of the flow procedure with the unit being $\mu$. The edge capacity used in {\em Unit-Flow} is $U$ units, i.e. $\mu U$, and the total source supply to {\em Unit-Flow} is at most $\frac{2m}{\mu}$ units, so the runtime of {\em Unit-Flow} is $O(\frac{mh}{\mu})$. As $\mu$ decreases geometrically starting with $F$, the total edge capacity used through the procedure is $2UF$, and the total runtime will be $O(\frac{mh}{\mu})$ for the $\mu$ at termination. 

The procedure will terminate either when each vertex $v$ gets at most $d(v)$ supply, which must happen once $\mu$ drops to $1$, or in some iteration $j$ we have
\begin{equation}
\label{eq:stop}
\vol(A_j)\geq \frac{\tau 2m}{10\mu\ln 2\mu \ln \ln m}.
\end{equation}
Note that the total amount of supply in any iteration is upper bounded by $|\Delta(\cdot)|=2m$, so {\em Unit-Flow} never finishes with case $(2)$ of Theorem~\ref{thm:unit-flow}. If {\em Unit-Flow} finishes with case $(1)$ in some iteration $j$, then $A_j=\emptyset$, and there is no excess on any node at the end of the termination, so no supply will be removed.

We need to argue (corresponding to the two cases in this Lemma respectively)
\begin{enumerate}
\item If we don't terminate early due to Eqn.~\eqref{eq:stop}, then at least $(1-\tau)$ fraction of the total source supply is routed to sinks.
\item If we terminate due to Eqn.~\eqref{eq:stop}, we must have $\vol(K)$ being  $\Omega(\frac{m}{F\ln m\ln\ln m})$, and the runtime in this case is $O(\vol(K)h\ln \frac{m}{\vol(K)}\ln\ln m)$, where $K$ is the side of $(A_j,\bar{A_j})$ with smaller volume. The conductance of the cut in this case follows from Theorem~\ref{thm:unit-flow} case $(3a),(3b)$.
\end{enumerate}

We first show case (1). 
By Theorem~\ref{thm:unit-flow} with $w=2$, the $A_j$ we get from {\em Unit-Flow} satisfies $2d(v)\geq f(v)\geq d(v)$ for all $v\in A_j$, and $f(v) \le  d(v)$ for all $v\in V\setminus A_j$. So we have $\ex_j (A_j)\leq \vol(A_j)\leq \frac{2m}{\mu_j}$ where we use $\mu_j$ to denote the value of $\mu$ in the iteration of $j$ and
$ex_j(\bar{A_j}) =0$. If we never have Eqn.~\eqref{eq:stop}, we know for all $j$
\[
\ex_j(A_j)\mu_j\leq \vol(A_j)\mu_j\leq \frac{\tau 2m}{10\ln 2\mu_j \ln\ln m}
\]
and if we add up the excess removed from all iterations, we have 
\[
\sum_j \ex_j(A_j)\mu_j\leq \frac{\tau m}{5\ln\ln m}(\sum_{j=0}^{\ln F}\frac{1}{j+1})\leq \tau 2m
\]
so the total supply remaining is at least $(1-\tau)2m$, and we have case (1). In this case when we terminate $\mu$ is at least $1$, so the runtime is $O(mh)$.

For case (2), let $j$ be the iteration when we get Eqn.~\eqref{eq:stop}, and we look at the cut $A_j$ returned. Let $K$ be the smaller side of $A_j$ and $\bar{A_j}$. The conductance of cut $(A_j,V\setminus A_j)$ follows from Theorem~\ref{thm:unit-flow} with $w=2$. We proceed to prove the runtime bound. We look at the two cases:
\begin{itemize}
\item $\mu_j\geq 2$ at termination: 
Note that $A_j$ is the smaller side of the cut in this case, i.e., $K= A_j$. Since $\vol(A_j)\leq \frac{2m}{\mu_j}$, we have $\mu_j\leq \frac{2m}{\vol(A_j)}$, thus we can rewrite Eqn.~\eqref{eq:stop} as
\[
\vol(A_j)\geq \frac{\tau 2m}{10\mu_j\ln \frac{4m}{\vol(A_j)} \ln \ln m}.
\] 
The running time is $O(\frac{m}{\mu_j}h)$, which is also $O(\vol(K)h\ln \frac{m}{\vol(K)}\ln\ln m)$.

\item $\mu_j=1$ at termination: If $\vol(A_j)\leq m$, $A_j$ is still the smaller side and we have the same argument as above. 
If, however, $\vol(A_j)\geq m$, we consider again two cases. If $\vol(A_j)\geq (1-\tau) 2m$, then at least a $(1-\tau)$ fraction of the total supply is routed to sinks and the running time is $O(mh)$, i.e. 
we are in Case $(1)$ of the lemma. If, however,  $ m \geq \vol(A_j) < (1-\tau) 2m$, then both $\vol(A_j)$,$\vol(V\setminus A_j)$ are $\Theta(m)$, since $\tau$ is a constant.
This implies that the running time, which is $O(mh)$, is $O(\vol({K})h)$.
\end{itemize}
In both cases, the running time is $O(\vol(K)h\ln \frac{m}{\vol(K)}\ln\ln m)$.

It remains to show $\vol({K})$ is $\Omega(\frac{m}{F\log m})$. From the above discussion, we know that if we end with case $(2)$ of the lemma, then either $\vol({A_j})\leq m$, or $\vol(\bar{A_j})\geq \frac{\tau}{1-\tau}\vol({A_j})$. Thus Eqn~\eqref{eq:stop} implies $\vol(K)$ is $\Omega(\frac{m}{F\ln m\ln\ln m})$.
\end{proof}

%% file: section2.tex
\vspace{-0.1in}
\newcommand{\ccc}{{1000}}
\section{The Kawarabayashi-Thorup decomposition framework}
\label{sxn:outerloop}
In the rest of the paper, we show how to modify the algorithm in~\cite{KT18} (the K-T algorithm) to work with $\alpha$-balanced directed graphs, and to use the efficient flow procedures in Section~\ref{sxn:flow}. Eventually we get a $O(m\alpha^4 \ln^2 m \ln\ln^2 m)$ algorithm for computing minimum cut in $\alpha$-balanced directed simple graphs. The structure of the algorithm (Algorithm~\ref{alg:KT}) is as follows: It consists of three nested loops,  which we call (a) the outer loop, (b) the middle loop, and (c)  the inner procedure. Both the algorithm of~\cite{KT18} and our algorithm use this structure, and we divide it  
into two layers in our discussion: the inner procedure where we replace with our own so it can work with our flow subroutine, and the K-T framework (i.e. everything outside the inner procedure) which we largely follow~\cite{KT18}, but extended to balanced directed graphs. We have a clean interface between the two, which is formally presented as Theorem~\ref{thm:middle-main}. We will discuss the K-T framework and the interface in this section. The modification of the K-T framework is only for the purpose of making it work on balanced directed graphs. We note both our extension of the K-T framework to balanced directed graphs and the extension to preserve approximate min cuts in~\cite{KT18} are straightforward adaptations of the original approach of~\cite{KT18} on min cuts.

The adaptation is mostly straightforward, and is similar to how~\ref{thm:middle-main} adapts their algorithm to also preserve approximate min-cuts. 
 
Given a $\alpha$-balanced directed simple graph $G=(V,E)$ with in-degree and out-degree lower bounded by $\delta$ (i.e., $\forall v\in V:\ind(v),\outd(v)\geq \delta$), the decomposition framework computes a (multi-)graph $\GG$ with ${O}(\frac{m_G\alpha^4 \ln m_G}{\delta})$ edges, while preserving all non-trivial minimum cuts of $G$. Note that $\delta$ upperbounds the value of the minimum cut, and when $\delta$ is $O(\alpha^4 \ln m_G)$, $G$ has $O(\frac{m_G\alpha^4 \ln m_G}{\delta})$ edges already. In this case, we run Gabow's algorithm directly on $G$, and it will be efficient. Otherwise, i.e. $\delta=\Omega(\alpha^4 \ln m_G)$, we construct and run Gabow's algorithm on the graph $\GG$.

The high-level approach is to start with $\GG=G$, and recursively contract subsets of nodes into {\em supervertices} to reduce the size of $\GG$ (the outer loop in Algorithm~\ref{alg:KT}), while preserving all non-trivial minimum cuts of $G$. Throughout the algorithm, a node (or vertex) in $\GG$ is either a regular vertex (i.e. a vertex in $G$) or a supervertex (i.e. a subset of vertices of $V$). We consider a supervertex to be both a node in $\GG$ and a subset of vertices of $V$, and supervertices can be further contracted with other nodes through the algorithm. 
At any point, the supervertices (as subsets of $V$) and the regular vertices (as singleton sets) in $\GG$ give a partition of $V$. All edges of $G$, except those whose both endpoints are collapsed into the same supervertex, are in $\GG$. In particular, the in-degree and out-degree of any regular vertex in $\GG$ are both at least $\delta$, and any cut in $\GG$ can be mapped back to a unique cut in $G$. A cut $(S,\bar{S})$ of $G$ survives in $\GG$ iff all edges in $E(S,\bar{S})$ still remain in $\GG$, that is, all the regular nodes contracted in a supervertex of $\GG$ lie entirely on one side of $(S,\bar{S})$. Through our algorithm, all non-trivial minimum cuts of $G$ will survive in $\GG$. 
\begin{algorithm}
\caption{Kawarabayashi-Thorup framework}
\label{alg:KT}
\fbox{\parbox{0.97\textwidth}
{
$\GG \leftarrow G$; $G$ has min degree $\delta\geq \Omega(\alpha^4 \ln m_G)$\\
{\bf Repeat (Outer loop)} \\
\tab $H \leftarrow \GG$, and $\tilde{H}$ denotes undirected version of $H$ .\\
\tab Remove passive supervertices from $\tilde{H}$ and trim $\tilde{H}$.\\
\tab {\bf While} {\em $\exists$ a connected component $C$ in $\tilde{H}$ is not a certified cluster} {\bf do (Middle loop)}\\
\tab \tab Let $s$ be the smallest integer such that $C$ is certified $s$-splittable.\\
\tab \tab {\bf Inner procedure:} Achieve either $(1)$, $(2)$, or $(3)$ (Theorem~\ref{thm:middle-main})\\
\tab \tab \tab $(1)$ Find low conductance cut ($A,\bar{A}$) of $C$ in time $\tilde O(\alpha^2\min(\vol_C(A), \vol_C(\bar{A}))$.\\
\tab \tab \tab $(2)$ Find low conductance cut ($A,\bar{A}$) of $C$ in time $\tilde O(\alpha^2 m_C)$,\\ 
\phantom{\tab \tab \tab} where $A$ is certified $\cc s$-splittable in $\GG$, and $\vol_C(A)$ is $\Theta(m_C)$.\\
\tab \tab \tab (3) Certify that $C$ is $\cc s$-splittable in time $\tilde O(\alpha^2 m_C)$. \\
\tab \tab {\bf End inner procedure}\\
\tab \tab  If a cut was found (case (1)(2)), remove the cut edges from $\tilde{H}$ and trim $\tilde{H}$. \\
\tab {\bf End while}\\
\tab Take each cluster component of $H$, and contract its core (if exist) to a supervertex in $\GG$.\\
\tab Contract any two vertices that have more than $\alpha \delta$ edges between them.\\
{\bf Until} {$\geq \frac{1}{20\alpha}$ of the edges in $\GG$ are incident to passive supervertices.}
}
}
\end{algorithm}
In each iteration of the outer loop, the algorithm computes (disjoint) subsets of nodes in $\GG$ that can be contracted. More specifically, we maintain $H$, with $H=\GG$ at the start of the iteration, edges and nodes will be removed through the iteration, and at the end, $H$ will be a collection of connected components such that each component must fall entirely on one side of any minimum cut, and, thus, can be contracted. Note that our algorithm will ignore the direction of the edges in $H$ until all connected component in $H$ (when treating $H$ as undirected) are {\em clusters} (to be defined shortly), i.e., when the middle loop in Algorithm~\ref{alg:KT} finishes. We denote $\tilde{H}$ as the undirected version of $H$, and refer to $\tilde{H}$ when the algorithm operates as if $H$ is undirected.

{\bf Passive supervertex and trimming.} At the start of the iteration, supervertices with (undirected) degree less than $c_1\alpha^2 \gamma \delta$ (called {\em passive} supervertices) are removed from $\tilde{H}$, where $c_1$ is a suitably chosen constant, and $\gamma=\Theta(\alpha\ln m_G)$. Throughout the iteration, whenever the algorithm removes edges and nodes from $\tilde{H}$, it will also {\em trim} $\tilde{H}$, which is to recursively remove from $\tilde{H}$ any vertex (regular or not) that has lost more than $\frac{3}{5}$ of its degree (comparing to its undirected degree in $\GG$). In particular, every connected component $C$ in $\tilde{H}$ will be {\em trimmed}, i.e. $\forall v\in C:d_C(v)=d_{H}(v)\geq \frac{2}{5}d_{\GG}(v)$.

Our goal is to detect subsets $C$ of nodes that can be contracted to supervertices. Thus for each such subset $C$ there cannot be a cut $S$ in $\GG$ such that $S$ maps back to a non-trivial minimum cut in $G$, and $|S \cap C|, |\bar{S} \cap C| \ge 1$. We call this condition on a set of nodes {\bf \em Requirement R1}. To find components that can be contracted, the algorithm first finds {\em clusters}.
\begin{definition}
\label{def:cluster}
A trimmed subset $C$ of $\tilde{H}$is a {\em cluster} if for any cut of undirected cut-size at most $\alpha \delta$ in $\GG$, at least one side of the cut contains less than $3\alpha$ regular vertices from $C$ and no supervertex from $C$. 
\end{definition}
When the middle loop in Algorithm~\ref{alg:KT} finishes, all connected components in $\tilde{H}$ are clusters, and we call them cluster components. A cluster component will almost entirely fall in one side of any directed minimum cut of $\GG$, and given a cluster component, it is easy to get its {\em core}, which is the part that can be contracted. 

\subsection{Construction of the core for a cluster} 
In this step the algorithm will work with $H$ (as directed graph). We now distill from each cluster component $C$ a subset of nodes that fullfils R1. To do so we 
consider $C$ to be a set of vertices and remove from it (and from $H$) all {\em loose} vertices, as well as edges incident to loose vertices.
\begin{definition}
\label{def:loose} 
A vertex $v$ in a cluster component $C$ is loose if it is a {\em regular} vertex, and {\em either} one of the following conditions holds.
\begin{enumerate}
\item $3\alpha + |E_{\GG}(\bar{C},v)| \geq |E_{\GG}(v,C)|-3\alpha$,
\item $3\alpha + |E_{\GG}(v,\bar{C})| \geq |E_{\GG}(C,v)|-3\alpha$.
\end{enumerate}
Recall that $E_{\GG}(v,C)$ denotes the set of edges (in $\GG$) directed from $v$ to nodes in $C$, and here $\bar{C} = \GG \setminus C$ is the set of nodes in $\GG$ but not in $C$.
\end{definition}

{\bf Shaving.} By {\em shaving} we refer to the operation of removing loose vertices, as well as their incident edges, from clusters (and also from $H$). As mentioned, a cluster component will almost entirely fall in one side of any directed minimum cut of $\GG$ (up to $3\alpha$ regular vertices can fall on the other side), and we will show (in Lemma~\ref{lemma:core}) that all these regular vertices that fall on the other side must be loose. Let $A\subset C$ be the set of nodes remaining in a cluster after shaving, we know $A$ fulfills R1, thus it can be contracted. However, we will only contract the nodes in $A$ to a supervertex if it has ``high enough'' internal volume.
\begin{definition}
\label{def:core}
For a cluster component $C$, let $A$ be the cluster $C$ after shaving, i.e., after removing loose vertices from $C$. We call $A$ a {\em core} of $C$ if the total number of edges in $\GG$ that are internal to $A$ (i.e., edges between nodes of $A$) is at least a $\frac{1}{\core}$ fraction of the total number of edges in $\GG$ that are incident to nodes of $C$.
\end{definition}  
{\bf Scraping and construction of the core.} If $A$ is not a core of the cluster $C$ then we remove all nodes in $A$ from $H$, which is called {\em scraping} $A$. After shaving and scraping, all that left in $H$ are cores of clusters, and we contract each core to a supervertex. This is the only way that supervertices are created.

\subsection{Finding cluster components}
As it is fairly simple to get the core from a cluster, the major work lies in finding cluster components.
%
We have the following measure of how close $C$ is to a cluster.
\begin{definition}
\label{def:splittability}
A connected component $C$ of $\tilde{H}$ is {\em $s$-splittable} if every cut $(S,\bar{S})$ of $\GG$ with undirected cut-size at most $\alpha\delta$ satisfy $\min(\vol_C(S\cap C),\vol_C(\bar{S}\cap C))\leq s$. We call $s$ the splittability of $C$.
\end{definition}
Informally, the smaller the splittability of $C$, the closer $C$ is to fall entirely in one side of any minimum cut of $\GG$ (as any directed minimum cut of $\GG$ will have undirected cut-size at most $\alpha\delta$). Note that a component $C$ is by definition $m_C$-splittable, and any subcomponent of a $s$-splittable component is also $s$-splittable. 
The strategy of the algorithm is to drive down the strengh of the connected components in $\tilde{H}$, and we will show (in Lemma~\ref{lem:cluster}) that a sufficient condition for a trimmed component $C$ to be a cluster is when its splittability is at most $s_0$ for some $s_0=\Theta(\delta\gamma\alpha^2)$.

To get components of smaller splittability, the algorithm relies on the inner procedure (See Theorem~\ref{thm:middle-main}). Each time the inner procedure is invoked, it is given a trimmed component $C$ in $H$ that is already certified to be $s$-splittable for some $s>s_0$. The inner procedure will either certify that $C$ is $\cc s$-splittable, or find a low conductance cut in $C$. In the latter case, we can remove the cut edges from $\tilde{H}$, and break $C$ into smaller components. This is useful since the volume of a component is a trivial bound on its splittability. The low conductance is crucial, as we need to bound the total number of cut edges removed during the entire process. Ultimately,  at least a constant fraction of the edges from $\GG$ will remain in the cores (of the cluster components) in $H$, which are contracted at the end, so the volume of $\GG$ drops geometrically across outer loop iterations (See Lemma~\ref{lemma:outer-halved} for details). 

Finally, when the cores are contracted and new supervertices are created, we can contract any vertices $u,v$ that have more than $\alpha\delta$ edges between them (the total of both directions). The reason is that if $u,v$ have more than $\alpha \delta$ undirected edge between them, then any cut in $G$ seperating $u,v$ must have undirected cut-size more than $\alpha\delta$, and by $\alpha$-balance, the cut must have directed cut-size more than $\delta$ in both directions, thus cannot be a minimum cut. Contracting these nodes will preserve all minimum cuts, and guarantee that at any point no pair of nodes in $\tilde{H}$ have more than $\alpha\delta$ parallel (undirected) edges between them. Note that for any two regular nodes in $\tilde{H}$, we further know that there can be at most $2$ parallel (undirected) edges, as the original directed graph is simple.

The runtime of each invocation of the inner procedure is proportional to the progress made in that invocation. That is, if it only finds a low conductance cut $(A, \bar A)$, the runtime is local (i.e., proportional to the size of $A$). If, however, the inner procedure spends $\tilde{O}(m_C\alpha^2)$ time on a component $C$, it certifies a smaller splittability for $C$ (or for a subcomponent of volume $\Theta(\vol(C))$ of $C$). See Section~\ref{sxn:runtime} for more details.

%% file: appendix-sec2.tex
\subsection{Analysis}
\label{app:section2}
In this section we show the correctness of the K-T framework (Algorithm~\ref{alg:KT}).
To prove the correctness  we show (1) the termination of the middle loop, (2) the termination of the outer loop, (3) that no non-trivial minimum cut of the original graph $G$ gets contracted into a supervertex, and (4) that the resulting graph $\GG$ contains $O(m_G\alpha^4\ln m_G/\delta)$ edges.

To guarantee the termination of the middle loop we have to show that at some point all connected components of $\tilde{H}$ are clusters or individual nodes. Starting with an $s$-splittable component $C$ in $\tilde{H}$,
each iteration of the middle loop either reduces the size of $C$ (and potentially shows that one of the new connected components is $\cc s$-splittable) or shows that $C$ is $\cc s$-splittable,
 i.e., either reducing the size of $C$ or its splittability. Note that removing edges of $\tilde{H}$ does not increase the splittability of its components, i.e., any component in $\tilde{H}$ that was $s$-splittable is also $s$-splittable after edge removal. Together with Lemma~\ref{lem:cluster} below, eventually every connected component in $\tilde{H}$ must either has size 1 or is a cluster, so the middle loop will always terminate.
\begin{mylemma}
\label{lem:cluster}
Let $s_0 = \ccc \alpha^2 \gamma \delta$, any trimmed $s_0$-splittable connected component $C$ in $H$ is a cluster.
\end{mylemma}
\begin{proof}
The proof crucially relies on the removal of passive supervertices and the trimming of $H$. 
Let $C$ be a trimmed $s_0$-splittable connected component of $\tilde{H}$, we have to show that 
for every cut in $\GG$ of undirected cut-size at most $\alpha\delta$,  one side contains
(a) no supervertex of $C$ and (b) less than $3\alpha$ regular vertices of $C$. Let $(S,\bar{S})$ be such a cut, and $B = C \cap S$, and $D = C \cap \bar{S}$. Since $C$ is an $s_0$-splittable connected component in $\tilde{H}$, we have that $\min (\vol_C(B), \vol_C(D)) \leq s_0$. We assume WLOG $\vol_C(B) \leq s_0$.

(a) We first show that $B$ contains no supervertex. As all passive supervertices were removed, and $C$ is trimmed, any remaining supervertex $v$ in $C$ has $d_C(v)\geq \frac{2}{5} c_1 \alpha^2\gamma \delta$. If $B$ contains $v$, we must have $\vol_C(B)\geq d_C(v)$, and as long as $s_0$ is chosen to be less than $\frac{2}{5} c_1 \alpha^2\gamma \delta$, we have a contradiction. Thus, our choice of a suitable $s_0=\Theta(\delta\gamma\alpha^2)$ guarantees that $B$ contains no supervertex from $C$.

(b) We argue next that $B$ contains less than $3\alpha$ regular vertices. Every regular node $v$ in $B$ has degree more than $4 \delta/5$ in $\tilde{H}$, since $d_{\GG}(v)=\ind_{\GG}(v)+\outd_{\GG}(v)$ is at least $2\delta$, and $v$ remains in the trimmed component $C$. Since all nodes in $B$ are regular nodes, there can be at most $2$ parallel (undirected) edges between any two nodes of $B$ (as our original directed graph is simple). Let $x$ be the number of vertices in $B$, then we know for any $v\in B$, at most $2x$ of the its incident edges can go to other vertices in $B$, thus at least $4 \delta/5-2x$ edges must go to vertices in $D$. As $(S,\bar{S})$ has cut-size at most $\alpha\delta$, the number of edges between nodes in $B$ and nodes in $D$ is also at most $\alpha\delta$, thus we have
\[
x(4 \delta/5-2x)\leq \alpha\delta.
\]
For contradiction, assume $x\geq 3\alpha$. From the above inequality, we must have $4 \delta/5-2x \leq \delta/3$, which gives $x\geq 7\delta/30$. Then $\vol_C(B)\geq (7\delta/30)\cdot (4\delta/5)$. Recall that $\delta=\Omega(\alpha^4 \ln m_G)$, and $s_0=\Theta(\alpha^2\gamma\delta)$ with $\gamma=\Theta(\alpha\ln m_G)$. With the right choice of constants, we have $(7\delta/30)\cdot (4\delta/5)>s_0$, which contradicts $C$ being $s_0$-splittable. Thus, we must have $x<3\alpha$, i.e., $B$ has less than $3\alpha$ regular vertices.
\end{proof}
The while loop (middle loop) terminates when all connected components remaining in $\tilde{H}$ are clusters, and as discussed earlier, the algorithm contracts the cores of the clusters into supervertices. The following lemma shows that every non-trivial minimum cut of the original graph $G$ survives the contraction of the cores.
\begin{mylemma}
\label{lemma:core}
If a non-trivial minimum cut of $G$ has survived in $\GG$ at the beginning of an iteration of the outer loop, then it will survive when we contract the core of any cluster component at the end of the iteration.
\end{mylemma}
\begin{proof}
Note that if a non-trivial min cut of $G$ survives in $\GG$, then it must still be a min cut $(S,\bar{S})$ of $\GG$. WLOG, assume the minimum cut is from $S$ to $\bar{S}$, then we have $|E_{\GG}(S,\bar{S})|\leq \delta$, and since $G$ is $\alpha$-balanced, $\GG$ is also $\alpha$-balanced, so the undirected cut-size of $S$ is at most $\alpha\delta$. Since the original min cut is non-trivial, we have that both $S$ and $\bar{S}$ have at least two regular nodes or one supervertex. 

Consider some core, $A$, and its associated cluster component $C$. Since $C$ is a cluster, we know one of $S$ or $\bar{S}$ has less than $3\alpha$ regular vertex and no supervertex from $C$. We consider the two cases.
\begin{enumerate}
\item If $C\cap S$ has at most $3\alpha$ regular vertices and no supervertex, we take any regular vertex $v\in C\cap S$. Consider moving $v$ from $S$ to $\bar{S}$, and how this affects the number of edges directed from $S$ to $\bar{S}$ (i.e., $E_{\GG}(S,\bar{S})$). Note since the original min cut is non-trivial, we know that $S$ and $\bar{S}$ both have at least two regular nodes or one supervertex. Thus, after moving the regular vertex $v$ from $S$ to $\bar{S}$, we still have a valid cut (i.e., $S$,$\bar{S}$ both non-empty).

When we move $v$ from $S$ to $\bar{S}$, all edges in $E_{\GG}(v,C\cap\bar{S})$ will be removed from $E_{\GG}(S,\bar{S})$, and there are at least $|E_{\GG}(v,C)|-3\alpha$ such edges, as $|C\cap S|\leq 3\alpha$. On the other hand, all edges from $\bar{C}\cap S$ to $v$ (at most $|E_{\GG}(\bar{C},v)|$ such edges), as well as those edges to $v$ from the (less than $3\alpha$) other regular vertices in $C\cap S$ will be added to $E_{\GG}(S,\bar{S})$. Thus, $|E_{\GG}(S,\bar{S})|$ will decrease by at least
\[
(|E_{\GG}(v,C)| - 3\alpha) - (|E_{\GG}(\bar{C},v)|+3\alpha).
\]
Since $(S,\bar{S})$ is a minimum cut in $\GG$, we must have $(|E_{\GG}(v,C)| - 3\alpha) \leq (|E_{\GG}(\bar{C},v)|+3\alpha)$, which means $v$ is loose, and will not remain in $A$.
\item The other case is when $C\cap \bar{S}$ has at most $3\alpha$ regular vertices and no supervertex. The argument is symmetric to the case above, and corresponds to the second condition in the definition of loose vertices (Definition~\ref{def:loose}).
\end{enumerate}
From the above discussion, we know that after removing loose vertices, all nodes remaining in $A$ must lie entirely in one side of $(S,\bar{S})$. Thus, the minimum cut will survive the contraction of $A$.
\end{proof}

The above argument actually shows that it is safe to contract the set of remaining nodes $A$ after shaving a cluster component $C$. However, the algorithm only contracts $A$ when it is a core, since otherwise we may end up with too many supervertices. Recall $A$ is a core when it has large internal volume, and formally it can be shown that every supervertex has $\Omega(\delta^2)$ volume contracted inside it. This will bound the total number of supervertices by $O(\frac{m_G}{\delta^2})$. Together with the degree bound on passive supervertices, we have the lemma below. 
\begin{mylemma}
\label{lem:passive}
The total number of edges in $\GG$ incident to passive supervertices
is $O(\frac{m_G\gamma\alpha^2}{\delta})$.
\end{mylemma}
\begin{proof}
If a supervertex $a$ was constructed by the contraction of a core $A$ we define the {\em volume of $a$} to be $\vol_G(A)$, i.e., the volume of the set of contracted regular nodes in the original graph $G$. 
We will show any supervertex has volume at least $ \delta^2/5$. Although a supervertex can be further contracted into a new supervertex, the new supervertex must have even larger volume. Thus, it suffices to prove the volume lower bound for the supervertices that only have regular nodes contracted inside. 

Consider a supervertex $a$ formed by contracting the core $A$ of some trimmed cluster component $C$ which only has regular nodes. Any regular node $v$ has $d_{\GG}(v)\geq 2\delta$, and if $v$ remains in the trimmed component $C$, we know $v$ has at least $2d_{\GG}(v)/5$ undirected edges to other regular nodes in $C$. Since there are at most two parallel undirected edges between any pair of regular nodes, we must have at least $2\delta/5$ regular nodes in $C$, which means the number of edges in $\GG$ incident to nodes of $C$ is at least $4\delta^2/5$. Since $A$ is a core of $C$, we know the total number of internal edges of $A$ is at least $\delta^2/5$. Thus, the volume of the supervertex $a$ (i.e. $\vol_{G}(A)$) is at least $\delta^2/5$.

For any passive supervertex, its degree in $\GG$ is at most $O(\alpha^2 \gamma \delta)$, which is at most a $O(\alpha^2\gamma /\delta)$ fraction of the volume contracted inside the passive supervertex. As the total volume of the original graph is $2m_G$, the total number of edges incident to passive supervertices is at most $O(m_G\alpha^2\gamma/\delta)$, which is $O(m_G\alpha^3\ln m_G/\delta)$. 
\end{proof}
This lemma, together with the termination condition that at least $\frac{1}{20\alpha}$ fraction of the edges in $\GG$ are incident to passive supervertices, gives that $\GG$ has $O(m_G\alpha^4\ln m_G/\delta)$ edges at termination.

So far we have shown Part (1),(3),(4) of the correctness proof, and what remains is the termination of the outer loop. For this purpose we will show that the number of edges in $\GG$ decreases geometrically in every iteration of the outer loop, while the number of edges in $\GG$ incident to
passive supervertices does not decrease as a supervertex, once it is passive, is always removed from $H$ at the beginning of the iteration, and thus, it is never contracted into a new supervertex. Thus, the outer loop terminates after
$O(\log m_G)$ iterations.

To show the reduction on number of edges in $\GG$, we proceed as follows.
Let $\bar m$ be the number of edges in $\GG$ at the begining of an iteration of the outer loop, and let
$\bar m'$ be the number of edges at the end of the iteration, which equals the number of edges at the beginning of the next iteration. Since all the edges in $H$ at the end of the iteration (i.e., those internal to the cores) are contracted, the edges remaining in $\GG$ at the end of the iteration are exactly those edges removed from $H$ through the iteration. In particular, these edges are of the following four types. 
\begin{enumerate}
\item were incident to a passive supervertex and removed from $H$ at the start of the iteration ({\em Type-1 edges}),
\item were trimmed from $H$ ({\em Type-2 edges}),
\item were on a  low-conductance cut found by the inner procedure, and removed from $H$ during the middle loop ({\em Type-3 edges}),
\item were shaved and scraped from a cluster to get the core({\em Type-4 edges}).
\end{enumerate}
Our goal is to show that there are at most $9 \bar m/10$ such edges (i.e., $\bar m' \le 9\bar m / 10$), which guarantees the number of edges in $\GG$ drop geometrically over iterations of the outer loop.

For this we proceed as follows. Let $\bar m_i$ denote the number of Type-$i$ edges for $i = 1, 2, 3, $ and $4$.
In Lemma~\ref{lem:cut-edge} we will show that $\bar m_1 + \bar m_3 \le \bar m/(10\alpha)$. We then prove in Lemma~\ref{lem:trim-shave} that  $\bar m_2 + \bar m_4 \leq 8\alpha(\bar m_1 + \bar m_3)$. Putting these results together we then conclude in Lemma~\ref{lemma:outer-halved}
that $\bar m' \leq  9 \bar m/10$, which is what we wanted to show.
To proceed we first show an auxiliary lemma.

\begin{mylemma}\label{lem:core2}
If a cluster component $C$ of $\tilde{H}$ has undirected cut-size $k$ in $\GG$ (i.e., $\partial_{\GG}(C)=k$), then there are at most $4\alpha k$ edges (in $\GG$) that are incident to nodes of $C$ but not internal to its core. 
\end{mylemma}
\begin{proof}
In this proof, the edges we consider are those in $\GG$. Let $T$ be the set of nodes shaved from $C$ (i.e., the loose regular vertices), let $A = C \setminus T$,  and let $l$ be the number of edges between $T$ and $\bar C=\GG \setminus C$ (in $\GG$). Then we know there are $k-l$ edges between $A$ and $\bar C$ (in $\GG$). 

We first bound the total number of edges incident to loose vertices (i.e., vertices in $T$). For a loose vertex $v$, there are two cases (corresponding to the two cases in the definition of loose vertex (Definition~\ref{def:loose}).
\begin{enumerate}
\item $v$ is loose due to $3\alpha + |E_{\GG}(\bar{C},v)| \geq |E_{\GG}(v,C)|-3\alpha$. There are two sub-cases here. (i) If $|E_{\GG}(v,\bar{C})|\geq \outd_{\GG}(v)/2$ (i.e., more than half of the outgoing edges of $v$ go to nodes not in $C$, and thus, are removed in $H$), then we know
\[
\ind_{\GG}(v)+\outd_{\GG}(v)\leq \alpha \outd_{\GG}(v)\leq 2\alpha|E_{\GG}(v,\bar{C})|.
\]
(ii) If $|E_{\GG}(v,\bar{C})|\leq \outd_{\GG}(v)/2$, then $|E_{\GG}(v,C)|\geq \outd_{\GG}(v)/2\geq \delta/2$. This, together with condition that $v$ is loose, gives
\[
|E_{\GG}(\bar{C},v)|\geq |E_{\GG}(v,C)| - 6\alpha \geq \frac{\outd_{\GG}(v)}{2}-6\alpha > \frac{\outd_{\GG}(v)}{3},
\]
where the last inequality is due to $\outd_{\GG}(v)\geq \delta > \alpha^4\ln m_G$. Thus, we have
\[
\ind_{\GG}(v)+\outd_{\GG}(v)\leq \alpha \outd_{\GG}(v)\leq 3\alpha|E_{\GG}(\bar{C},v)|
\]
Put the bounds we get from cases (i) and (ii) together, we know in both cases, 
\[
d_{\GG}(v)\leq 3\alpha(|E_{\GG}(\bar{C},v)|+ |E_{\GG}(v,\bar{C})|) 
\]
\item $v$ is loose due to $3\alpha + |E_{\GG}(v,\bar{C})| \geq |E_{\GG}(C,v)|-3\alpha$. By a symmetric argument of the case above, we can also conclude
\[
d_{\GG}(v)\leq 3\alpha(|E_{\GG}(\bar{C},v)|+ |E_{\GG}(v,\bar{C})|) 
\]
\end{enumerate}
Thus, for any loose vertex $v$, we know $d_{\GG}(v)$ is at most a $3\alpha$ factor larger than the number of edges between $v$ and nodes not in $C$. Thus, the total volume of the loose vertices is $\vol_{\GG}(T)\leq 3\alpha l$, which means that the total number of edges incident to nodes of $C$ but not internal to $A$ (i.e., not between two nodes of $A$) is at most $3\alpha l + k- l \leq 3\alpha k$. This proves the lemma unless the core is empty (i.e, $A$ is scraped). 

If the core is empty, we know the total number of edges internal to $A$ is at most $1/\core$ of the edges incident to $C$, so the total number of edges internal to $A$ is at most $1/3$ of the number of edges incident to $C$ but not internal to $A$. This suggests that at most $\alpha k$ edges are internal to $A$, and thus at most $4\alpha k$ edges are incident to $C$. All the edges incident to $C$ are not internal to its core, which is empty, so there are at most $4\alpha k$ such edges. This completes the proof of the lemma.
\end{proof}


\begin{mylemma}\label{lem:trim-shave}
In each iteration of the outer loop $\bar m_2 + \bar m_4 \leq 8\alpha(\bar m_1 + \bar m_3)$.
\end{mylemma}
\begin{proof}
Consider an iteration of the outer loop, note that the outer loop first performs an intermixed sequence of cutting and trimming operations, followed by a sequence of shaving and scraping operations, each processing one cluster. Note further that
whether a node is shaved or not does {\em not} depend on its degree in $H$, but only on its degree in ${\GG}$ and on the
partition of nodes into clusters and non-cluster nodes. It also does not depend on the other scraping operations. Thus in the analysis we can assume that all shaving operations happen {\em before} all the 
scraping operations.

We use an amortization argument. Let the {\em lost degree} $dl(v)$ of a node $v$ in $H$ be $d_{\GG}(v) - d_{\tilde{H}}(v)$, and we look at the total lost degree of the nodes remaining in $H$, which is denoted as $dl(H)$. We start with $dl(H)=0$ when $H=\GG$, when we remove edges from $H$, $dl(H)$ will increase, but when we remove nodes from $H$, $dl(H)$ will decrease (as we no longer count the lost degree of those removed nodes). For every Type-$1$ edge that is removed, $dl(H)$ increases by $1$, since the degree of the endpoint other than the passive supervertex decreases by $1$ (the passive supervertex won't be in $H$, so its degree is not counted in the total lost degree). The removal of every Type-$3$ edge increases $dl(H)$ by $2$. Now if we trim a node $v$, $dl(v)$ is at least $3d_{\GG}(v)/5$, which will be subtracted from $dl(H)$ as $v$ is no longer in $H$, but as we remove the remaining at most $2d_{\GG}(v)/5$ edges of $v$, we increase $dl(H)$ by $2d_{\GG}(v)/5$, and in total, the trimming of a vertex $v$ reduces $dl(H)$ by at least $d_{\GG}(v)/5$. Trimming a vertex $v$ removes at most $2d_{\GG}(v)/5$ Type-$2$ edges, and $dl(H)$ decreases by at least $d_{\GG}(v)/5$. Thus, the number of Type-$2$ edges is at most twice the decrease in $dl(H)$. 

From our discussion above, we see that after all the cutting and trimming, i.e., the removal of all Type-$1$,$2$,$3$ edges, if the total lost degree of nodes in $H$ is $dl(H)$, then we have
\begin{equation}
\label{eqn:trimmed}
0\leq \bar m_2 \leq 2(\bar m_1 + 2\bar m_3 - dl(H))  .
\end{equation}
Now we know that before shaving and scrapping, the total lost degree of nodes remaining in $H$ is $dl(H)$, which is exactly the sum of $\partial_{\GG}(C)$ over all the cluster components $C$ in $H$, since all that remain in $H$ are the cluster components. We can then apply Lemma~\ref{lem:core2} to each cluster component, and conclude that in total at most $4\alpha dl(H)$ edges in $\GG$ that are incident to nodes of the cluster components in $H$, but not internal to the final cores. Among these at most $4\alpha dl(H)$ edges, $dl(H)$ are already removed before shaving and scrapping by the definition of lost degree, so the total number of edges shaved or scraped (i.e. $\bar m_4$) is at most $(4\alpha -1 )dl(H)$. Together with our earlier bound on $\bar m_2$ (Eqn.~\eqref{eqn:trimmed}), we get
\begin{align*}
\bar m_2 + \bar m_4 \leq &2(\bar m_1 + 2\bar m_3 - dl(H))+(4\alpha -1 )dl(H)\\
\leq  & 4(\bar m_1+\bar m_3) + (4\alpha -3)dl(H)\\
\leq & 4(\bar m_1+\bar m_3) + (4\alpha -3)\cdot 2(\bar m_1+\bar m_3)\\
\leq & 8\alpha (\bar m_1+\bar m_3)
\end{align*}
where the second to last inequality is because $dl(H)\leq \bar m_1 + 2\bar m_3$ from Eqn.~\eqref{eqn:trimmed}.
\end{proof}

\begin{mylemma}
\label{lem:cut-edge}
For each but the last iteration of the outer loop it holds that
$\bar m_1 + \bar m_3 \le \frac{\bar m}{10\alpha}$.
\end{mylemma}
\begin{proof}
Note all the argument in this proof are treating the graph as undirected graph. As this iteration is not the last iteration of the outer loop, the termination condition is not met at the end of the previous iteration, which suggests that at most $\bar m/(20\alpha)$ edges are incident to passive supervertices at the beginning of this iteration, i.e., $\bar m_1 \le \bar m/(20\alpha)$. Thus, it remains to bound $\bar m_3$, i.e., the number of edges removed on the low conductance cuts found during the inner procedure (Theorem~\ref{thm:middle-main}).

Whenever we cut a connected component $C$ of $\tilde{H}$ into $A$ and $C\setminus A$, with $A$ being the smaller side, we know $\Phi_C(A)\leq \frac{(\log m_C-\ceil*{\log \vol_C(A)})}{20\alpha\log m_G}$. Thus we charge $\frac{(\log m_C-\ceil*{\log \vol_C(A)})}{20\alpha\log m_G}$ to each edge (in $\tilde{H}$) incident to $A$ to account for the Type-$3$ edges removed across the cut. 
Let $A_1, \dots, A_l$ be the smaller side of the cut to which an arbitrary edge $e_0$ belongs whenever it is charged, and define $A_0$ to be the graph $\tilde{H}$ right before the first iteration of the middle loop (i.e., after removing the passive supervertices and trimming). That is, for every $i=1,\ldots, l$, we have $A_{i-1}$ as the connected component $C$ being chosen for an iteration of the middle loop, and we find a low conductance cut, with $A_{i}$ being the smaller side, and the edge $e_0$ falls internal to $A_{i}$. Thus, the total charge that $e_0$ receives during all executions of the middle loop in this iteration of the outer loop is
$$\sum_{1 \le i \le l} \frac{\ceil*{(\log m_{A_{i-1}}}-\ceil*{\log \vol_{A_{i-1}}(A_i)})}{20\alpha\log m_G} \le 
\frac{\log m_{A_0}}{20\alpha \log m_G} \le
\frac{1}{20\alpha}.$$
This shows that there are at most $\bar m/(20\alpha)$ type-3 edges. Thus, $\bar m_1 + \bar m_3 \le \bar m/(10\alpha)$.
\end{proof}

Putting these two lemmata together we conclude that the number of edges in $\GG$ is reduced by a constant
factor in each iteration of the outer loop.

\begin{mylemma}
\label{lemma:outer-halved}
In each except for the last iteration of the outer loop, i.e. the repeat loop, the number of
edges in the graph $\GG$ is decreased by a factor of at least $0.9$. 
\end{mylemma}
\begin{proof}
Lemma~\ref{lem:trim-shave} showed that  $\bar m_2 + \bar m_4 \leq 8\alpha(\bar m_1 + \bar m_3)$.
Lemma~\ref{lem:cut-edge} showed that  $\bar m_1  + \bar m_3 \leq \bar m/(10\alpha)$.
Thus it follow that $\bar m'= \bar m_1 + \bar m_2 + \bar m_3 + \bar m_4 \le 9 \bar m /10$.
\end{proof}

The above lemma gives the termination of the outer loop, and completes the analysis of the K-T framework. We summarize the results of this subsection in the following theorem.

\begin{theorem}
Given a simple $\alpha$-balanced directed graph $G=(V,E)$ such that $\forall v\in V: \ind(v),\outd(v)\geq \delta$, we can compute a multi-graph $\GG =(\bar{V}, \bar{E})$ where $\GG$ has $O(m_G\alpha^4\ln m_G / \delta)$ edges, and all non-trivial (directed) minimum cuts of $G$ are preserved in $\GG$.
\end{theorem}

%% file: section3.tex
\vspace{-0.1in}
\section{The inner procedure}
\label{sxn:inner-procedure}
We start with a high level descriptions of the inner procedure, followed by the detailed algorithm and analysis in Section~\ref{app:inner-procedure}. We follow a similar approach as~\cite{KT18}, but use the flow methods in Section~\ref{sxn:flow} instead of diffusions as subroutines. For the inner procedure, the graph we consider will always be the undirected version. Moreover, as the inner procedure is given a connected component $C$ of $\tilde{H}$, within the scope of this section, we can consider the induced sub-graph on $C$ as our graph.
As discussed in Section~\ref{sxn:outerloop}, the K-T framework relies on the inner procedure to achieve the following.
\begin{theorem}
\label{thm:middle-main}
Given an $s$-splittable trimmed component $C$ with $m_C\geq s\geq s_0=\ccc\alpha^2\gamma\delta$, the inner procedure will achieve one of the following:
\begin{enumerate}
\item Find a set $A\subset C$ with $\vol_C(A)\leq m_C$, and 
\[
\Phi_C(A)\leq \frac{(\log m_C+1-\ceil*{\log \vol_C(A)})}{20\alpha\log m_G}
\]
in time
$O(\alpha^2\vol_C(A)\ln \frac{m_C}{\vol_C(A)}\ln m_G \ln\ln^2 m_G)$.
\item Find a set $A\subset C$ (and $\bar{A}=C\setminus A$) with 
\[
\Phi_C(A)\leq \frac{(\log m_C+1-\ceil*{\min(\log \vol_C(A),\log\vol_C(\bar{A}))})}{20\alpha\log m_G}
\]
in time $O(\alpha^2 m_C\ln m_G\ln\ln m_G)$. Moreover, $\vol_C(A)$ is $\Theta(m_C)$, and $A$ is certified to be $\cc s$-splittable.
\item Certify that $C$ is $\cc s$-splittable in time $O(\alpha^2 m_C\ln m_G\ln\ln m_G)$.
\end{enumerate}
\end{theorem}
The intuition is as follows. If $C$ is a connected component of $H$ that is not a cluster, by definition there must exist cuts in $C$ of cut-size at most $\alpha\delta$. Consider any such small cut
and denote by $S$ the side with minimum volume.
We know $\vol_C(S)\leq s$, since $C$ is $s$-splittable. The cut-size being at most $\alpha \delta$ gives a strong bottleneck to route into or out of $S$, and we can exploit this bottleneck.\\

\smallskip
\noindent{\bf Use the cut as a bottleneck to route supply out of $S$.}
The excess-scaling algorithm (Algorithm~\ref{alg:scaling-flow}) guarantees to find a low conductance cut in local runtime if we give it a very infeasible flow problem, i.e.~one where it is impossible to route a large fraction of the source supply to sinks (Lemma~\ref{lemma:flow-alg}). 
A small cut $S$ naturally gives such  very infeasible flow problems as follows. As the total sink capacity in $S$ is $\vol_C(S)$, the condition $\vol_C(S)\leq s$ 
bounds the total sink capacity in $S$ by at most $s$. As there are at most $\alpha \delta$ edges on the cut, we can pick an appropriate edge capacity parameter to get a good bound on the total cut capacity of $S$. As long as we choose a source function such that the source supply in $S$ is large, for example twice the sum of $S$'s sink and cut capacity, we get a very infeasible flow problem. 
The difficulty, however,  is  to construct such a source function {\em without knowing $S$}. The strategy, very informally, is as follows. We construct a large number ($\leq 5000\alpha$) of flow problems with different source functions, and run (in parallel, step by step) Algorithm~\ref{alg:scaling-flow} on them,
terminating them whenever one of them returns a low conductance cut, or, if this does not happen, letting them all run to termination. If any of these flow problems had large enough source supply in $S$, we  get a low conductance cut in local runtime, i.e., case $(1)$ in Theorem~\ref{thm:middle-main}. \\

\smallskip
\noindent{\bf Use the cut as a bottleneck to route supply into $S$.}
If we do not get the above case, we end with case $(1)$ of
Lemma~\ref{lemma:flow-alg} for all the flow problems we constructed,
and we have spent ${O}(\alpha^2 m_C \ln m_G\ln\ln m_G)$ time for them. In this case, we
know that the source functions of these flow problems all have little
source supply starting in $S$ as they were able to route most of their
flow to a sink. Using the $\Delta(v)'$ values returned by each
execution of Algorithm~\ref{alg:scaling-flow}, 
we suitably combine the successfully routed source supplies to 
a new, well spread-out source function. 
More specifically, the new source function
fulfills the following properties:
(a) Very little source supply is in $S$, and the cut bounds the amount of supply that can be pushed into $S$, so the total supply ending in $S$ must be small. (b) The amount of total supply is large (more formally at least $4m_C$) and well spread out (more formally $\forall v:\Delta(v)\leq 25d(v)$).
Thus we can run a  {\em Unit-Flow} computation directly on it ({\em without} going through the excess scaling procedure).  
We use $h = \Theta(\alpha \ln m_G \ln \ln m_G)$, $U=s/(20\alpha \delta)$, and $w = 25$, and have either of the two outcomes below.

(A) All nodes in $C$ have their sinks saturated (case $(2)$ of Theorem~\ref{thm:unit-flow}). Since the amount of supply ending in $S$ is small, the total sink capacity in $S$ must also be small, i.e. $\vol_C(S)$ must be small. Recall $S$ is any cut in $C$ with cut-size at most $\alpha \delta$. Thus we know all such cuts have small volume, more specifically at most $\cc s$, implying that $C$ is $\cc s$-splittable, i.e. case $(3)$ of Theorem~\ref{thm:middle-main}.

(B) We get a set $A$ as specified in case $(3b)$ of Theorem~\ref{thm:unit-flow}. Since all nodes in $A$ have their sinks saturated, by a similar argument as above, we show that $\vol_C(A\cap S)$ is small. Again as this argument works for any $S$ of cut-size at most $\alpha\delta$, we can argue $A$ is $\cc s$-splittable.
Additionally, case (3b) of Theorem~\ref{thm:unit-flow} and Observation~\ref{obs:largevol} give the desired bound on
 $\Phi_C(A)$ and show that $\vol_C(A) = \Theta(m_C)$, which shows that case $(2)$ of Theorem~\ref{thm:middle-main} holds.

Note that the outline of the flow problem construction is similar
to the seeding of diffusions in \cite{KT18}, but the details differ in
part due to their ability to use the linearity property of diffusions.
We must explicitly spread out our flows and warm start our procedures
in some cases as noted above.

%% file: appendix-inner.tex
\subsection{Details and analysis of the inner procedure}
\label{app:inner-procedure}
Now we formalize the intuition we outlined above, and give the detailed implementation and analysis of the inner procedure. 

Recall when the inner procedure is invoked, we have a connected component $C$ of $\tilde{H}$ such that $C$ is certified to be $s$-splittable for some $s\in [s_0,m_C]$. Through the rest of the section, we work completely inside $C$ (the induced sub-graph of $C$ as an undirected graph), and the volume, degree, cut-size are all defined with $C$ being our graph when we omit the subscript. 

Recall from Section~\ref{sxn:outerloop} the K-T framework removes passive supervertices from $\tilde{H}$, and keeps $\tilde{H}$ trimmed through the algorithm. Thus in the connected component $C$, any regular vertex $v$ has $d(v)\geq \frac{2}{5}2\delta$, and any supervertex has degree at least $\frac{2}{5}c_1\alpha\delta\gamma$ where $\gamma=\Theta(\alpha\ln m)$ is the parameter in the definition of a passive supervertex. Moreover, no two vertices in $C$ have more than $\alpha \delta$ parallel edges between them, since such pair of vertices would have been contracted at the end of the previous iteration of the outer loop, and any two regular vertices in $C$ have at most $2$ parallel edges between them, since the original graph $G$ is a simple directed graph. 

As discussed in Section~\ref{sxn:outerloop}, we need to prove Theorem~\ref{thm:middle-main}, and we will follow the intuitions outlined at the start of Section~\ref{sxn:inner-procedure}.

As suggested earlier, we will construct various flow problems aiming to exploit the existence of a cut $S$ of cut-size at most $\alpha\delta$, and use the flow-based algorithms from Section~\ref{sxn:flow} on the constructed problems. The edge capacity parameter is crucial if we want to use the cut as a bottleneck to route supply out of or into $S$. As specified in Lemma~\ref{lemma:flow-alg}, when given parameter $U$, the actual edge capacities used by the algorithm is $UF$ where $F=\max_v\frac{\Delta(v)}{d(v)}$ is the largest ratio between a node's initial supply and its degree, thus it is important for us to construct flow problems where the source function $\Delta$ has small $F$. Formally, our strategy to construct source function with small $F$ is captured in the following definitions.
\begin{definition}
An {\em edge-bundle} is a set of edges sharing a common endpoint. We denote an edge-bundle by $(v,X(v))$, where $v$ is the common endpoint that we call the {\em center} of the edge-bundle, and $X(v)$ is the multiset (as there are parallel edges) containing the other endpoints of the edges. A set of edge-bundles are disjoint if their underlying sets of edges are edge disjoint. 
\end{definition}
Note that in a set of disjoint edge-bundles, a vertex $v$ can still be the center of multiple edge-bundles, and we can also have parallel edges, the definition simply prevents the same edge from appearing in multiple edge-bundles.
\begin{definition}
\label{def:expansion-graph}
Given a set $Y$ of edge-bundles in $C=(V,E)$, the {\em expansion graph} associated with $Y$ is the directed multigraph $G_Y=(V,E_Y)$, such that $E_Y$ has a directed edge $(v,u)$ for each $u\in X(v)$ and each $(v,X(v))\in Y$. Namely, $E_Y$ is the union of all edge-bundles in $Y$, with edges oriented away from the centers of the edge-bundles.
\end{definition}
\begin{definition}
\label{def:sparse-set}
A set of edge-bundles $Y$ in $C$ is {\em $(\beta,Z)$-sparse} if
\begin{itemize}
\item The edge-bundles in $Y$ are edge disjoint.
\item Each edge-bundle $(v,X(v))\in Y$ has at least $Z$ edges.
\item For each vertex $v$, its in-degree in the associated expansion graph $G_Y$ is at most $\frac{1}{\beta}$ of its degree in $C$.
\end{itemize}
\end{definition}
Note if $Y$ is $(\beta,Z)$-sparse, then any subset of $Y$ is also $(\beta,Z)$-sparse.

Edge-bundles will be used to construct source functions for flow problems, and the motivation of $(\beta,Z)$-sparse set of edge-bundles is that if we put supply on the centers of the edge-bundles in the set, and push out uniformly using the edges in the edge-bundles, the amount of supply received by any node will not be too large comparing to its degree. 

 More precisely, we call an {\em initial spread-out} of $\sigma$ supply over the edge-bundle $(v,X(v))$ as the operation of starting with $\sigma$ supply on $v$, and pushing $\frac{\sigma}{|X(v)|}$ supply along each edge in the edge-bundle to vertices in $X(v)$. Formally, given edge-bundle $(v,X(v))$ and $\sigma$, we define 
\[
\Delta_{(v,X(v)),\sigma}(u)=\frac{\sigma\cdot \#(u,X(v))}{|X(v)|}
\]
where $\#(u,X(v))$ is the number of times $u$ appears in $X(v)$. That is, $\Delta_{(v,X(v)),\sigma}(u)$ is the supply ending at $u$ if we start with $\sigma$ supply at the center $v$ of the edge-bundle $(v,X(v))$, and then push out all the $\sigma$ supply at $v$ evenly along the edges in the edge-bundle. We extend the definition to a set $Y$ of edge-bundles:
\[
\Delta_{Y,\sigma}(u)=\sum_{(v,X(v))\in Y}\Delta_{(v,X(v)),\sigma}(u)
\]
i.e. $\Delta_{Y,\sigma}(u)$ is the amount of supply ending at $u$, if we carry out simultaneously a initial spread-out of $\sigma$ supply over each edge-bundle in $Y$. It is clear that the total amount of supply is $|\Delta_{Y,\sigma}(\cdot)|=|Y|\sigma$, where $|Y|$ is the number of edge-bundles in $Y$.

We will use the supply on vertices arising from initial spread-outs as the source function, and we consider flow problems defined below.
\begin{definition}
\label{def:source-flow}
Given $\Delta:V\rightarrow \mathbb{Z}_{\geq 0}$ and $\kappa$, we define a flow problem, {\em Flow-Problem}$(\Delta,\kappa)$, as follows. The source function is given by $\Delta(\cdot)$, all edges have capacity $\kappa$, and each vertex $v$ is a sink of capacity $d(v)$.
\end{definition}
Essentially we are taking a two-phase approach to spread supply from edge-bundle centers to the entire graph. The first phase being the initial spread-outs, where we have full control of the behavior, and the second phase being the flow routing, so we can still take advantage of the better conductance property of flow algorithms.

The flow algorithm we use in Section~\ref{sxn:unit-flow} and Section~\ref{sxn:flow-algorithm} will also allow us to associate each unit of supply with its source vertex as specified by $\Delta(\cdot)$. When the $\Delta(\cdot)$ we use arises from initial spread-outs over edge-bundles, we can further decompose the flow to associate each unit of supply with the original edge-bundle it started at, i.e. before the initial spread-out. Thus in step~\ref{inner:s-free}, we assume we know how much of the supply originating from each edge-bundle is routed to sinks.
\noindent
\begin{algorithm}
\caption{Inner Procedure}
\label{alg:inner}
\fbox{
\parbox{0.95\textwidth}{
{\bf Input:} Trimmed component $C$ with $m_C$ edges, and $s\in [s_0,m_C]$ such that $C$ is $s$-splittable.\\
{\bf Steps:}
\begin{enumerate}[leftmargin=*]
\item
\label{inner:start}
Choose a set $Y$ of $\frac{5000\alpha m_C}{s}$ edge-bundles that is $(\alpha \gamma,\frac{\delta}{10})$-sparse, and split $Y$ into sets $Y_1,\ldots,Y_{5000\alpha}$, each with $\frac{m_C}{s}$ edge-bundles.
\item
\label{inner:cut}
In parallel (step by step) for all $i=1,\ldots,5000\alpha$, solve Flow-Problem$(\Delta_{Y_i,2s},\frac{s}{1000\alpha \delta} )$ using Algorithm~\ref{alg:scaling-flow} in Section~\ref{sxn:flow-algorithm}, with inputs graph $C$, source function $\Delta_{Y_i,2s}$, $\tau=0.1$, $U=100\alpha\ln m_G$, and $h=1000\alpha\ln m_G\ln\ln m_G$. Terminate all problems if the Flow-Problem of any $i$ terminates with a cut $A$ as in case $(2)$ of Lemma~\ref{lemma:flow-alg}, stop the inner procedure with $A$.
\item
\label{inner:s-free}
Otherwise, the Flow-Problems for all $i$ end with case $(1)$ of Lemma~\ref{lemma:flow-alg}, i.e. with at least $1.8m$ supply routed to sinks. For each Flow-Problem $i$, use the returned preflow $f_i$ to find a subset $X_i\subseteq Y_i$ such that each edge-bundle in $X_i$ has at least $1.6s$ of its $2s$ initial supply routed to sinks. 
\item
\label{inner:process}
For each $i$, compute $g_i(\cdot)$ from $f_i(\cdot)$ as follows: First remove from each $f_i(v)$ the excess supply on vertices (i.e. $\max(f_i(v)-d(v),0)$ supply on $v$), as well as the supply not originating from edge-bundles in $X_i$. Then scale the remaining supply at every vertex by $\frac{1}{200\alpha}$. 
\item
\label{inner:certify}
Let $\Delta_X(\cdot)\myeq \sum_i g_i(\cdot)$. Run {\em Unit-Flow} in Section~\ref{sxn:unit-flow} with inputs $G=C$, source function $\Delta_X$, $U=\frac{s}{20\alpha\delta}$, $h=1000\alpha\ln m_G\ln\ln m_G$, and $w=25$. If the returned preflow routes at least $d(v)$ supply to every vertex $v$, stop and output that $C$ is $\cc s$-splittable. Otherwise, stop with the set $A$ returned by {\em Unit-Flow}, and output that $A$ is $\cc s$-splittable.
\end{enumerate}
}
}
\end{algorithm}
We use the following definition to formally specify whether an edge-bundle is ``inside'' or ``outside'' a small cut.
\begin{definition}
\label{def:s-capture}
An edge-bundle $(v,X(v))$ is {\em $s$-captured} if there exists a cut $S$ in $C$ such that $\partial(S)\leq \alpha\delta$, $s_0\leq \vol(S) \leq s$, and $|X(v)\cap S|\geq \frac{3}{4} |X(v)|$, i.e. at least $\frac{3}{4}$ of the edges are between $v$ and vertices in $S$. We say the edge-bundle is $s$-captured by $S$. (Note that $v$ might or might not belong to $S$.) A {\em $s$-free} edge-bundle is one that is not $s$-captured.
\end{definition}

In Step~\ref{inner:start} of the inner procedure, we pick a large set $Y$ of edge-bundles that is $(\alpha\gamma,\frac{\delta}{10})$-sparse. This step is valid as we have the following lemma.
\begin{restatable}{mylemma}{sparseset}
For $s_0 \ge 1000 \alpha^2 \gamma \delta$, a trimmed component $C=(V,E)$ with $m=|E|$, and $m \geq s\geq s_0$, we can construct a set of $\frac{5000\alpha m}{s}$ edge-bundles that is $(\alpha \gamma,\frac{\delta}{10})$-sparse. The construction takes $O(\frac{m\alpha^2 \delta\gamma}{s})$.
\end{restatable}
\begin{proof}
Let $Z=\frac{\delta}{10}$, in our construction, we say a supervertex is {\em live} if it has at least $\alpha \gamma Z$ edges to live neighbors, and a regular vertex is {\em live} if it has at least $Z$ edges to live neighbors. Vertices are {\em dead} if not {\em live}. We will implicitly consider a graph $C'$ on live vertices. As $C$ being a trimmed component, we know at the start all regular vertices have degree at least $8Z$, and all supervertices have degree at least $8\alpha^2 \gamma Z$ (as passive supervertices are removed), so we can make all vertices live at the start, and $C'=C$. As $\delta\gg \alpha \gamma$ in our setting, for simplicity we assume integrality of $\frac{d_{C'}(v)}{\alpha \gamma}$ for any live vertex $v$, as $d_{C'}(v)\geq Z=\frac{\delta}{10}$. 

Now we describe how we construct the set of edge-bundles $Y$ with the following procedure.\\
\noindent
\fbox{
\parbox{0.95\textwidth}{
{\bf Construction of edge-bundles}
\begin{enumerate}[leftmargin=*]
\item Choose an arbitrary live supervertex, if no live supervertex exists, choose a live regular vertex. Call the chosen vertex $v$.
\label{proc:one}
\item Construct an edge-bundle centered at $v$ by picking $Z$ incident edges of $v$ in $C'$, subject to the constraint that for each live neighbor $u$ of $v$, we pick at most $\min(d_{C'}(v,u),\frac{d_{C'}(u)}{\alpha \gamma})$ parallel edges between $u$ and $v$, where $d_{C'}(v,u)$ is the number of edges between $v$ and $u$ in $C'$. Add the edge-bundle to $Y$.
\label{proc:add}
\item Remove edges from $C'$ as follows
\begin{enumerate}[label={(\alph*)}]
\item For each edge $\{v,u\}$ added to the edge-bundle above, remove that edge and an additional $\alpha\gamma-1$ incident edges of $u$ from $C'$.
\item Recursively remove from $C'$ the dead vertices and all their incident edges.
\end{enumerate}
\label{proc:remove}
\item
Repeat the process from Step~\ref{proc:one} until we have $\frac{5000\alpha m}{s}$ edge-bundles in $Y$.
\end{enumerate}
}
}

First we show that Step~\ref{proc:add} is always feasible, i.e. we can obtain such an edge-bundle with a live vertex $v$. As all vertices in $C'$ are live, if $v$ is a supervertex, the number of edges we can pick is 
\[
\sum_u\min(d_{C'}(v,u),\frac{d_{C'}(u)}{\alpha \gamma})\geq \sum_u\frac{d_{C'}(v,u)}{\alpha \gamma}\geq \frac{d_{C'}(v)}{\alpha \gamma}\geq Z
\]
If $v$ is a regular vertex, $C'$ must have no supervertex at that point, so there are at most two parallel edges between $v,u$ in $C'$. In this case, we can add any incident edge $(v,u)$ to the edge-bundle of $v$ as $d_{C'}(u)/(\alpha\gamma)\geq Z/(\alpha\gamma)\gg 2$, and $v$ has at least $Z$ incident edges in $C'$.

The condition we enforce in Step~\ref{proc:add} guarantees that we can always carry out Step~\ref{proc:remove}$(a)$, i.e. there will be enough edges to remove. If we have added $k$ edge-bundles to $Y$, the total number of edges we removed in Step~\ref{proc:remove}$(a)$ is $k Z\alpha \gamma$. To bound the number of edges removed in Step~\ref{proc:remove}$(b)$,
assume that every removal in Step~\ref{proc:add}  and Step~\ref{proc:remove}(a) places one token on the other endpoint of the removed edge. Thus, a total of $k Z\alpha \gamma$ tokens are placed on nodes.
We will show that we can consume one token on each edge removed during Step~\ref{proc:remove}(b), which bounds the number of edges removed in Step~\ref{proc:remove}(b) by $k Z \alpha \gamma$ (i.e., the total number of tokens created by \textbf{Step~\ref{proc:remove}(a)}) .
Whenever a dead vertex $v$ is removed, it consumes one token on each removed adjacent edge and it moves one token to the other endpoints of this edge. 
It remains to show that $v$ has a sufficient number of tokens to do so. We show this by proving a more general claim by induction: at each point in time, the number of tokens placed on a vertex corresponds to the number of edges the vertex has lost. 
This guarantees that each dead vertex we remove has enough tokens to give to its removed edges and its neighbors, since such a vertex has lost at least $3/4$-th of its adjacent edges.
The claim certainly holds before and when the first dead vertex is removed as  it received a token for all its previously removed edges. Next consider the removal of the $i$-th dead vertex $v$ and assume by induction that right before the removal $v$ has at least $3d_C(v)/4$ many tokens. As $v$ has lost at least $3/4$-th of its edges, the removal of $v$ 
removes at most $d_C(v)/4$ many edges. We use $d_C(v)/4$ many of $v$'s tokens and give them to the removed edges and another $d_C(v)/4$ many tokens and give them to the other endpoints of the removed edges. Thus, the induction invariant also holds after the removal of the $i$-th dead vertex.
This bounds the total number of edge removed in Step~\ref{proc:remove}$(b)$ by the total number of edges removed in Step ~\ref{proc:add} and~\ref{proc:remove}$(a)$, and thus after adding $k$ edge-bundles to $Y$, we have removed at most $2kZ\alpha \gamma$ edges from $C$ to get the current $C'$.

As long as $C'$ is not empty, it guarantees a live vertex, and thus an edge-bundle to add. We showed that the total number of edges removed from $C'$ after constructing $k$ edge-bundles is $2kZ\alpha \gamma$, thus as long as $2kZ\alpha \gamma\leq m$, we must have edges remaining in $C'$. This implies that we can have at least $\frac{m}{2Z \alpha \gamma}\geq \frac{5m}{\alpha \delta\gamma}$ edge-bundles, so as long as $s\geq s_0\geq 1000\alpha^2 \delta\gamma$, we can have $\frac{5000\alpha m}{s}$ edge-bundles.

The set of edge-bundles is clearly $(\alpha \gamma,\frac{\delta}{10})$-sparse by our construction as for each edge that we added to an edge-bundle, and that will become an in-edge for a vertex $v$ we removed $\alpha \gamma-1$ edges incident to $v$. As to the runtime, since we implicitly keep $C'$, the work is linear in the total number of edges removed from $C$, which is $2Z\alpha \gamma$ per edge-bundle, thus $O(\frac{m\alpha^2 \delta\gamma}{s})$ in total.
\end{proof}
First we show that if the procedure terminates at Step~\ref{inner:cut}, we get a cut as specified in case $(1)$ of Theorem~\ref{thm:middle-main}.
\begin{mylemma}
\label{lemma:inner-cut}
If the inner procedure stops at Step~\ref{inner:cut}, we have a set $A$ such that $\vol_C(A)\leq m_C$, and $\Phi(A)\leq \frac{(\log m_C + 1 -\ceil*{\log \vol_C(A)})}{20\alpha \log m_G}$. The running time in this case is $O(\vol(A)\alpha^2 \ln\frac{m_C}{\vol_C(A)}\ln m_G\ln\ln^2 m_G)$. 
\end{mylemma}
\begin{proof}
If at Step~\ref{inner:cut} the excess scaling flow algorithm terminates with case $(2)$ of Lemma~\ref{lemma:flow-alg} for any of the $5000\alpha$ flow problems, let $A$ be the smaller side of the cut returned. We know $A$ has the desired conductance, as $|\Delta_{Y_i,2s}(\cdot)| = \frac{m_C}{s}\cdot 2s = 2m_C$, $h=1000\alpha \ln m_G\ln\ln m_G$ and $U=100\alpha \ln m_G$. As to the runtime, since we run all $O(\alpha)$ flow problems in parallel, the time we spend before we terminate with $A$ is $O(\vol_C(A)\alpha^2\ln\frac{m_C}{\vol_C(A)}\ln m_G\ln\ln^2 m_G)$ by Lemma~\ref{lemma:flow-alg}. Furthermore, Lemma~\ref{lemma:flow-alg} guarantees that $\vol_C(A)$ is $\Omega(\frac{m_C}{F\ln m_G\ln\ln m_G})$, where $F=\max_v \frac{\Delta_{Y_i,2s}(v)}{2d(v)}$. 

We now upper-bound the value of $F$ in the Flow-Problems associated with the $Y_i$'s. Since each $Y_i$ is $(\alpha\gamma,\frac{\delta}{10})$-sparse, we know $\Delta_{Y_i,2s}(v)\leq \frac{2s}{\delta/10}\frac{d(v)}{\alpha \gamma}$ for all $v$ by Definition~\ref{def:sparse-set} and the construction of $\Delta_{Y_i,2s}$ from initial spread-outs. Thus $F\leq \frac{10s}{\alpha\delta\gamma}$, which implies $\vol_C(A)$ is $\Omega(\frac{m_C\alpha\delta\gamma}{s\ln m_G\ln\ln m_G})$. To find $Y$ in the first step of the inner procedure, we spend time $O(\frac{m_C\alpha^2 \delta\gamma}{s})$, which is $O(\alpha\vol_C(A)\ln m_G\ln\ln m_G)$. Thus the total runtime is $O(\vol_C(A)\alpha^2\ln\frac{m_C}{\vol_C(A)}\ln m_G\ln\ln^2 m_G)$ if the inner procedure ends in Step~\ref{inner:cut}.
\end{proof}

Now we formalize the intuition that if the source function has large enough initial source supply trapped inside the small cut, we get a very infeasible flow problem. 
\begin{mylemma}
\label{lemma:single-source}
Given an $s$-captured edge-bundle $(v,X(v))$, we can send at most $1.6s$ supply to sinks in {\em Flow-Problem}$(\Delta_{(v,X(v)),2s},\frac{s}{1000\alpha\delta})$.
\end{mylemma}
\begin{proof}
The flow problem we consider is with source function resulting from an initial spread-out of $2s$ supply over an edge-bundle $(v,X(v))$, which is $s$-captured by some set $S$. We know at least $\frac{3}{4}$ of the edges go to vertices in $S$, so after the initial spread-out, the total source supply at vertices outside $S$ is at most $\frac{s}{2}$. As vertices in $S$ have total sink capacity $\vol(S)$, which is at most $s$, the total amount of supply that can be routed to sinks in $S$ is at most $s$. Furthermore, at most $\frac{s}{1000\alpha\delta}\alpha\delta$ supply can be pushed out of $S$, since the cut has size at most $\alpha\delta$, and edges have capacity $\frac{s}{1000\alpha \delta}$. Even if all the $\frac{s}{2}+\frac{s}{1000}\leq 0.6s$ supply not in $S$ is routed to sinks eventually, we have at most $1.6s$ supply routed to sinks in total.
\end{proof}

Given the above lemma, if the Flow-Problems associated with all $Y_i$'s successfully route most of the supply to sinks, we know many of the edge-bundles we start with are not $s$-captured.
\begin{mylemma}
\label{lemma:s-free}
If $\gamma > c_2\alpha \ln m_G$ for a suitably chosen constant $c_2$, Step \ref{inner:s-free} of inner procedure will have at least $\frac{m\alpha}{10s}$ edge-bundles in each $X_i$, and all the edge-bundles in $X_i$ are $s$-free.
\end{mylemma}
\begin{proof}
Lemma~\ref{lemma:single-source} states that if an edge-bundle is $s$-captured, after an initial spread-out of $2s$ supply over the edge-bundle, at most $1.6s$ of the $2s$ supply can be subsequently routed to sinks, as long as the flow respects the edge capacity of $\frac{s}{1000\alpha\delta}$ on each edge. By Lemma~\ref{lemma:flow-alg}, we know the edge capacity used in the excess scaling flow algorithm is $200F\alpha\ln m_G$ (since we use $U=100\alpha \ln m_G$), where by our calculation in Lemma~\ref{lemma:inner-cut} we have $F\leq \frac{10s}{\alpha\delta\gamma}$. Thus the preflow $f_i$ respects the edge capacity of $\frac{2000s\alpha\ln m_G}{\alpha \delta\gamma}$, which is less than $\frac{s}{1000\alpha\delta}$ when $\gamma>c_2\alpha\ln m_G$ for some constant $c_2$. By Lemma~\ref{lemma:single-source}, any $s$-captured edge-bundle in $Y_i$ has at most $1.6s$ of its initial supply routed to sinks, so we know all edge-bundles in $X_i$ are $s$-free. 

To bound the size of $X_i$, note that we have $2s$ supply starting with each of the $\frac{m}{s}$ edge-bundles, thus if less than $\frac{m}{10s}$ of them have more than $1.6s$ rounted to sinks, we can have at most $\frac{m}{10s}2s+\frac{9m}{10s}1.6s< 1.8m$ total supply routed to sinks. Since at least $1.8m$ supply is routed to sinks, we must have at least $\frac{m}{10s}$ edge-bundles in each $X_i$.
\end{proof}
If we get to Step~\ref{inner:s-free} of the inner procedure, we have spent $O(\alpha^2 m_C \ln m_G\ln\ln m_G)$ on all the flow problems in the earlier step, so we need to make progress by certifying that a subset of volume $\Omega(m_C)$ is $\cc s$-splittable. We now formalize the strategy we outlined in the second half of high-level discussion of Section~\ref{sxn:inner-procedure}.

Intuitively, we want to continue with the pre-flows we have from Step~\ref{inner:cut}, since we know from these pre-flows that we can spread out the supply of all edge-bundles in $X$. However, if we simply start from scratch on edge-bundles in $X$, we may end up with some small cut, because flow routing is not a linear operator. The procedure we carry out in Step~\ref{inner:process} is mainly to get around the non-linearity of flow routing, so we can essentially keep the work done in the earlier step on spreading out the supply of edge-bundles in $X$.

In Step~\ref{inner:process}, we get preflow $g_i$ from the preflow $f_i$ for each $i$, and the union $\sum_i g_i$ is also a preflow. This preflow must be source-feasible with respect to a implicit source function (i.e. if we reverse the preflow $\sum_i g_i$) which we call $\Delta_0$ and that we only use for the analysis and do not need to compute. It is different from the $\Delta_X$ in Step~\ref{inner:certify}:  If we start with source function $\Delta_0$, and route according to $\sum_i g_i$, we would have $\Delta_X(v)$ supply ending at $v$. Note in the actual algorithm, we only need to compute $g_i$'s as the supply ending at vertices, i.e. $g_i(v)$'s, but not the actual routing, i.e. $g_i(u,v)$'s, as long as we know there is a valid routing that ends with the $g_i(v)$'s. 

\begin{mylemma}
\label{lemma:D0}
$50m_C\geq |\Delta_0(\cdot)| = |\Delta_X(\cdot)| \ge 4m_C$.
\end{mylemma}
\begin{proof}
By construction $\Delta_0(\cdot)$ is the source function of a preflow, and $\Delta_X(\cdot)$ is the supply ending at vertices after the preflow. Thus $|\Delta_0(\cdot)| = |\Delta_X(\cdot)|$. 

In Step~\ref{inner:process}, the amount of supply that $g_i$ keeps from $f_i$ is $\frac{1}{200\alpha}$ fraction of the non-excess supply originating from any edge-bundle in $X=\cup_{i=1}^{5000\alpha} X_i$. As any edge-bundle in $X$ has at least $1.6s$ supply routed to sinks, and we have at least $\frac{5000\alpha m_C}{10s}$ edge-bundles in $X$ by Lemma~\ref{lemma:s-free}, in total we keep at least $\frac{1.6s}{200\alpha}\frac{5000\alpha m_C}{10s}=4m_C$ supply. By construction, $\Delta_X$ has all this supply, i.e., $|\Delta_X(\cdot)| \ge 4m_C$.

The upperbound of $50m_C$ is because each $g_i$ has at most $\frac{2m_C}{200\alpha}$ total supply by scaling $f_i$, and the $5000\alpha $ groups in total make it at most $50m_C$ total supply in $\Delta_X(\cdot)$.
\end{proof}

Now we define a flow problem $\Pi$, where the source function is $\Delta_0$, each vertex $v$ is a sink of capacity $d(v)$, and edges have capacity $U = \frac{s}{10\alpha \delta}$. We first analyse $\Pi$ and then use it to analyze Step 5 in the subsequent lemma.
\begin{mylemma}
\label{lemma:s-strong}
For any feasible preflow of $\Pi$, the set $B$ of vertices with their sink capacities saturated, is $\cc s$-splittable.
\end{mylemma}
\begin{proof}
Consider any cut $S$ such that $\partial(S)\leq \alpha \delta$, $s_0\leq \vol(S) \leq s$. We will bound the total amount of supply that can end in $S$ for any feasible preflow of $\Pi$. 

We first look at the amount of supply that starts in $S$. The preflow $f_i$ starts with source function $\Delta_{Y_i,2s}$, so if we examine our construction of $g_i$ from $f_i$, we can mimic the changes on $\Delta_{Y_i,2s}$ to obtain the source function of $g_i$. Thus, the source function $\Delta_0$ can be obtained equivalently as follows: (i) We start with $\frac{2s}{200\alpha}$ supply (corresponding to the scaling) at the center of each edge-bundle in $X$ (corresponding to only keeping supply orginating from edge-bundles in $X$); (ii) carry out the initial spread-outs; (iii) and then remove some supply (corresponding to the removal of excess supply in $f_i(\cdot)$). Then it is clear we can bound the amount of supply that $\Delta_0$ has in $S$ by the amount of supply that would have been in $S$ without the Step (iii).

Since all edge-bundles in $X$ are $s$-free, if any edge-bundle has its center $v$ in $S$, at least $\frac{1}{4}$ of its $\frac{\delta}{10}$ edges cross $(S,\bar{S})$. As the cut-size is $\alpha\delta$, among all edge-bundles in $X$, at most $40\alpha$ of them have their centers in $S$, which means at most $\frac{2s}{200\alpha}\cdot 40\alpha=0.4s$ supply can be in $S$ before all the initial spread-outs. Since $X$ is a subset of $(\alpha \gamma,\frac{\delta}{10})$-sparse set $Y$, $X$ is also $(\alpha \gamma,\frac{\delta}{10})$-sparse. Thus the initial spread-outs push at most $\frac{2s}{200\alpha}/(\frac{\delta}{10})=\frac{s}{10\alpha \delta}$ supply along each edge. Thus, as the cut-size of $S$ is $\alpha \delta$, at most an additional $\frac{s}{10\alpha \delta} \cdot \alpha \delta = {\frac{s}{10}}$ supply can end in $S$ after the initial spread-outs. Thus in total, the source function $\Delta_0$ can have at most $0.4s+0.1s=0.5s$ supply starting in $S$.

Subsequently any valid preflow pushes at most $\frac{s}{10\alpha \delta}$ supply along each edge due to the edge capacity constraints in $\Pi$, so an additional $\frac{s}{10}$ supply can be routed into $S$ by the preflow. In total that means at most $\cc s$ supply can end in any such set $S$.

Now consider $B$, the set of all $v$ that receives at least $d(v)$ supply. We must have $\vol_C(B\cap S)\leq \cc s$ for any set $S$ such that $\partial(S)\leq \alpha \delta$, $s_0\leq \vol(S) \leq s$. This is enough to certify that $B$ is $\cc s$-splittable, since $B$ is already $s$-splittable as a sub-component of $C$ (so no need to consider any $S$ with $\vol(S)>s$).

This is equivalent to the definition of $B$ (as an induced sub-component) being $\cc s$-splittable, as we are working inside a connected component $C$ of $H$, so $vol_B(B\cap S)\leq \vol_C(B\cap S)$.
\end{proof}

We now finish the proof of Theorem~\ref{thm:middle-main} by showing the following lemma.
\begin{mylemma}
Step~\ref{inner:certify} will in time $O(m_C \alpha \ln m_G\ln \ln m_G)$ either certify that the entire component $C$ is $\cc s$-splittable, i.e. Case $(3)$ of Theorem~\ref{thm:middle-main}, or find a subset $A$ with 
\[\Phi(A)\leq \frac{(\log m_C+1-\ceil*{\min(\log \vol_C(A),\log\vol_C(C\setminus A))})}{20\alpha \log m_G}.
\]
In the latter case $A$ has volume $\Omega(m_C)$, and is certified to be $\cc s$-splittable, i.e.,~Case $(2)$ of Theorem~\ref{thm:middle-main}.
\end{mylemma}
\begin{proof}
In Step~\ref{inner:certify} of the inner procedure we run {\em Unit-Flow} of Section~\ref{sxn:unit-flow} with inputs $G=C$, source function $\Delta_X$, $U=\frac{s}{20\alpha\delta}$, $h=1000\alpha\ln m_G\ln\ln m_G$, and $w=25$. Note that it fulfills the assumptions on inputs of {\em Unit-Flow} in Theorem~\ref{thm:unit-flow} as $h\gg \ln m_G$, $w \ge 2$ and,
by the construction of $\Delta_X$, which removes all excess supply from all preflows $f_i$, it holds that $\Delta_X(v)\leq \frac{d(v)}{200\alpha}5000\alpha=25d(v) = w d(v)$ for all $v$. Let $f$ be the pre-flow returned by the {\em Unit-Flow} invocation.

Recall $\sum_i g_i$ is source-feasible with respect to the source function $\Delta_0(\cdot)$, and by routing according to $\sum_i g_i$, the supply ending at each vertex $v$ is $\Delta_X(v)$. Essentially $f$ resumes the routing by having $\Delta_X(\cdot)$ as source-function, so we can piece $\sum_i g_i$ and $f$ together, and obtain a preflow $f^*$ that is source-feasible with respect to $\Delta_0(\cdot)$.

To show $f^*$ is a feasible preflow for the flow problem $\Pi$, we need to further show $f^*$ respects the $\frac{s}{10\alpha\delta}$ edge capacity of $\Pi$. From Step~\ref{inner:cut} of the inner procedure, we have each preflow $f_i$ using at most $\frac{s}{1000
\alpha \delta}$ edge capacity, thus by construction $\sum_i g_i$ uses edge capacity of at most $\frac{s}{1000\alpha\delta}\frac{5000\alpha}{200\alpha}=\frac{s}{40\alpha\delta}$. As in the {\em Unit-Flow} invocation we use edge capacity $U=\frac{s}{20\alpha \delta}$, the preflow $f^*$, as a union of $\sum_i g_i$ and $f$, routes at most $\frac{3s}{40\alpha\delta}$ supply on each edge, and is thus a feasible preflow of $\Pi$.

As $f^*$ is a valid preflow of $\Pi$, we know from Lemma~\ref{lemma:s-strong} the set $B$, containing all vertices $v$ receiving at least $d(v)$ supply, is $\cc s$-splittable, and so is any subset of $B$. Since $f^*$ appends $f$ after $\sum_i g_i$, the supply ending at vertices is given by $f(\cdot)$, so $B=\{v|f(v)\geq d(v)\}$.

In Lemma~\ref{lemma:D0} we showed that the total supply $|\Delta_X(\cdot)|$ for Step~\ref{inner:certify} is at least $4m_C$. As all the vertices can absorb only $2m_C$ supply in total, our invocation of {\em Unit-Flow} won't end with case (1) of Theorem~\ref{thm:unit-flow}. If $f$ returned by {\em Unit-Flow} fulfills Case $(2)$ of Theorem~\ref{thm:unit-flow}, i.e. all vertices get at least $d(v)$ supply, we are guaranteed that $C=B$ is $\cc s$-splittable. 

On the other hand if {\em Unit-Flow} returns a set $A$ as in Case $(4)$ of Theorem~\ref{thm:unit-flow}, we show 
(a) $\Phi(A)\leq \frac{(\log m_C+1-\ceil*{\min(\log \vol_C(A),\log\vol_C(C\setminus A))})}{20\alpha \log m_G}$,  (b) $A$ is certified to be $\cc s$-splittable, and (c) $\vol_C(A)$ is $\Omega(m_C)$. This implies that $A$ satisfies the conditions of Case $(2)$ of Theorem~\ref{thm:middle-main}.

(a) The property of conductance follows directly from Case $3(b)$ of Theorem~\ref{thm:unit-flow}: As we use $h=1000\alpha\ln m_G\ln\ln m_G, w=25, U=\frac{s}{20\alpha\delta}\geq \frac{s_0}{20\alpha \delta}\geq \alpha \gamma\geq c_2\alpha^2 \ln m_G$.
(b) 
We know any vertex $v\in A$ receives at least $d(v)$ supply, so  $A \subseteq B$ is $\cc s$-splittable.
(c) As any vertex $v\notin A$ receives at most $d(v)$ supply, any vertex $v\in A$ receives at most $25d(v)$ supply, and since there is at least $4m_C$ total supply, we must have $25 \vol_C(A) + (2m_C - \vol_C(A)) \ge 4m_C$, which implies that $\vol_C(A)\geq  \frac{2}{24} m_C = \Omega(m_C)$. Thus $A$ satisfies all the conditions and the proof of the lemma is complete.
\end{proof}
The runtime of case (2) and (3) of Theorem~\ref{thm:middle-main} is $O(m_C\alpha \ln m_G\ln\ln m_G)$, as that's the total runnning time of the steps involved.

%% file: runtime.tex
\vspace{-0.1in}
\section{Running time analysis}
\label{sxn:runtime}
\begin{theorem}
The minimum cut in a $\alpha$-balanced simple directed graph with $m$ edges can be computed in time $O(\alpha^2 m \ln^{2}m\ln\ln^2 m+\alpha^4 m \ln^{2}m)$.
\end{theorem}
\begin{proof}
To compute the min cut of a $\alpha$-balanced simple directed graph $G$ with $m_G$ edges, we first construct $\GG$ as discussed earlier, and use Gabow's min-cut algorithm~\cite{Gabow91} on $\GG$. 
We start with the runtime to construct $\GG$. Recall that we use the K-T framework (Algorithm~\ref{alg:KT}) with our flow based inner procedure (Algorithm~\ref{alg:inner}). By Lemma~\ref{lemma:outer-halved}, $m_{\GG}$ decreases geometrically across iterations of the outer loop. As the runtime of each outer loop iteration will be $\Omega(m_{\GG})$, the first iteration will dominate asymptotically, so we focus on the first iteration, with $m_{\GG}=m_G$.

The operations outside of the middle loop in total take $O(m_{\GG})$ time. To analyze the middle loop, we look at each invocation of the inner loop. Informally we will charge the runtime to edges such that an edge is charged when it lies in the smaller side of a cut, or the splittability of its component drops by a constant factor. More specifically, given an $s$-splittable component $C$ in $H$, we have three cases by Theorem~\ref{thm:middle-main}.

\begin{enumerate}[label={(\arabic*)}]
\item Find a cut $(A,C\setminus A)$ with $\vol_C(A)\leq m_C$ in time 
\[
O(\alpha^2\vol_C(A)\ln ({m_C}/{\vol_C(A)})\ln m_G\ln\ln^2 m_G).
\] We can charge $O(\alpha^2\ln ({\vol_C(C)}/{\vol_C(A)}) \ln m_G\ln\ln^2 m_G)$ to each edge in $A$.

Consider the total charge to any edge by all invocations of inner procedure of this case. The edge is charged when it falls in the smaller side of a cut. The $\ln\frac{\vol(C)}{\vol(A)}$ part will telescope, so in total each edge is charged $O(\alpha^2\ln^2 m_G \ln \ln^2 m_G)$. 
\item Find a subset $A$ in $C$ where $\vol_C(A)$ is $\Theta(m_C)$, and $A$ is certified to be $\cc s$-splittable. The runtime is $O(\alpha^2 m_C\ln m_G\ln\ln m_G)$. We can charge 

$O(\alpha^2\ln m_G\ln\ln m_G)$ to each edge in $A$.

Over all invocations of inner procedure of this case, any edge is charged at most $O(\ln m_{\GG})$ times, since the splittability of its component decreases geometrically each time we charge the edge. In total each edge is charged $O(\ln^2 m_G\ln\ln m_G)$.
\item Certify the entire component $C$ is $\cc s$-splittable. The time we spend in this case is $O(\alpha^2 m_C\ln m_G\ln\ln m_G)$. We use the same argument as in case (2) above.
\end{enumerate}
In total, we can charge the runtime of the middle loop to the edges in $\GG$, and each edge is charged $O(\alpha^2\ln^2 m_G\ln\ln^2 m_G)$, so the runtime is $O(\alpha^2 m_G\ln^2 m_G\ln\ln^2 m_G)$.

At the end, we get a multi-graph $\GG$ with $O(\frac{\alpha^4 m_G\ln m_G}{\delta})$ edges, preserving all non-trivial min cuts of $G$. We use Gabow's min-cut algorithm~\cite{Gabow91} on $\GG$. Gabow's algorithm works on directed multi-graphs, and takes time $O(\lambda m_{\GG} \ln m_{\GG})$ on $\GG$, where $\lambda$ is the size of the min cut. With our bound on $m_{\GG}$, as well as $\lambda \le \delta$, the runtime of Gabow's algorithm is thus $O(\alpha^4 m_G\ln^2 m_G)$. Together with the runtime to construct $\GG$, we get the stated runtime bound in the theorem.
\end{proof}